\documentclass[11pt]{article}

\usepackage[paper=letterpaper, margin=1in]{geometry}

\usepackage[T1]{fontenc}
\usepackage[utf8]{inputenc}
\usepackage{csquotes}
\usepackage{authblk}
\usepackage[authoryear,round]{natbib}
\usepackage{fancybox}
\usepackage{array,float}
\usepackage{amsmath}
\usepackage{amssymb}
\usepackage{amsthm}
\usepackage{ifthen}
\usepackage{graphicx,color,xcolor}
\usepackage{enumitem}
\definecolor{cornellred}{rgb}{0.7, 0.11, 0.11}
\definecolor{dgreen}{rgb}{0.0, 0.5, 0.0}
\definecolor{royalazure}{rgb}{0.0, 0.22, 0.66}
\usepackage{url}
\usepackage{hyperref}
\hypersetup{
    colorlinks = true,
    linkcolor=cornellred,
    citecolor=royalazure,
    linkbordercolor = {white}
}
\usepackage[capitalize]{cleveref}
\usepackage{nicefrac}       % compact symbols for 1/2, etc.
\usepackage{microtype}      % microtypography

\usepackage{booktabs} % For formal tables

\usepackage[all]{xy}
\usepackage{mathrsfs}
\usepackage[noend,ruled]{algorithm2e} % For algorithms

\usepackage{bm} % For bold symbols that are used as vectors

\ifthenelse{\isundefined{\ACMSigConfMultiLineSplit}}
{
\def\OriginalDisplay{}
\def\OriginalDisplayEnd{}
\long\def\MultiLineSplit#1\MultiLineSplitEnd{}
}
{
\def\MultiLineSplit{}
\def\MultiLineSplitEnd{}
\long\def\OriginalDisplay#1\OriginalDisplayEnd{}
}

\newtheorem{theorem}{Theorem}[section]
\newtheorem{lemma}[theorem]{Lemma}

\newtheorem{corollary}[theorem]{Corollary}
\newtheorem{proposition}[theorem]{Proposition}

\newtheorem*{corollary*}{Corollary}

\theoremstyle{definition}
\newtheorem{definition}[theorem]{Definition}

\newtheorem{remark}[theorem]{Remark}

\newcommand\PARA[1]{\paragraph{#1}}

\newcommand\FPS[2]{{s}^{#2}({#1})}

\newcommand\FPA[2]{{x}^{#2}({#1})}

\newcommand\VEC[1]{\bm{#1}}

\newcommand\ChoicePayment[2]{\hat{q}_{#1}^{\, #2}}

\newcommand\FAVitems[2]{\hat{A}_{#1}^{#2}}

\newcommand\FPremain[2]{\VEC{\SurplusSym}_{#1}^{#2}}

\newcommand\FPprice[2]{\VEC{p}_{#1}^{#2}}

\newcommand\Multiplicity{capacity}

\newcommand\MultiSym[1]{c(#1)}
\newcommand\MultiSymB[1]{c\big(#1 \big)} %B is for Big
\newcommand\MaxMultiSym{C}
\newcommand\Surplus{inventory}

\newcommand\SurplusSym{\lambda}
\newcommand\Item{resource}
\newcommand\Items{resources}
\newcommand\Buyer{user}
\newcommand\Buyers{users}

\newcommand\Slot{slot}
\newcommand\Slots{slots}
\newcommand\LeftNode{i}
\newcommand\RightNode{t}
\newcommand\LeftSet{\mathtt{Left}}
\newcommand\RightSet{\mathtt{Right}}
\newcommand\LeftInts{[N]}
\newcommand\RightInts{[T]}

\def\NoEpsilon{}
\ifthenelse{\isundefined{\NoEpsilon}}
{
\def\EpsilonDependency{}
\def\EpsilonDependencyEnd{}
\long\def\ConstantDependency#1\ConstantDependencyEnd{}
}
{
\def\ConstantDependency{}
\def\ConstantDependencyEnd{}
\long\def\EpsilonDependency#1\EpsilonDependencyEnd{}
}

\newcommand\Step{round}
\newcommand\Steps{rounds}

\newcommand{\Capacity}{\mathit{c}}
\newcommand{\Active}{\mathit{A}}
\newcommand{\Valuation}{\mathit{v}}
\newcommand{\PriceVec}{\bm{p}}
\newcommand{\Pay}{\mathit{q}}

\newcommand{\SurplusVec}{\bm{\lambda}}
\newcommand{\PricePol}{\gamma}
\newcommand{\Selected}{\hat{\Active}}

\newcommand{\Ex}{\mathbb{E}}

\newcommand{\PricePols}{\Gamma}

\newcommand{\LVEC}{\big\langle}
\newcommand{\RVEC}{\big\rangle}

\DeclareMathOperator*{\ArgMax}{argmax}

\newcommand\MechanismName{\texttt{LBPP}}

\DeclareMathOperator{\argmax}{argmax}
\DeclareMathOperator{\argmin}{argmin}
\newcommand{\states}{\mathcal{S}}
\newcommand{\actions}{\mathcal{X}}

\newcommand{\tinit}{t_{\textrm{init}}}
\newcommand{\tfinal}{t_{\textrm{final}}}

\newcommand{\sinit}{s_{\textrm{init}}}

\title{Stateful Posted Pricing with Vanishing Regret via Dynamic Deterministic Markov Decision Processes}

\author[1]{Yuval Emek}
\author[1]{Ron Lavi}
\author[2]{Rad Niazadeh}
\author[1]{Yangguang Shi}

\affil[1]{Faculty of Industrial Engineering and Management, Technion, Haifa, Israel}
\affil[2]{University of Chicago Booth School of Business, Chicago, IL, United States}
\affil[ ]{\texttt{\{yemek, ronlavi, shiyangguang\}@ie.technion.ac.il}\\\texttt{rad.niazadeh@chicagobooth.edu}}

\date{}

\begin{document}

\maketitle

%\begin{abstract}
\begin{abstract}
In this paper, a rather general online problem called \emph{dynamic resource allocation with capacity constraints (DRACC)} is introduced and studied in the realm of posted price mechanisms.
This problem subsumes several applications of stateful pricing, including but not limited to posted prices for online job scheduling and matching over a dynamic bipartite graph.
As the existing online learning techniques do not yield vanishing-regret mechanisms for this problem, we develop a novel online learning framework defined over deterministic Markov decision processes with \emph{dynamic} state transition and reward functions.
We then prove that if the Markov decision process is guaranteed to admit an oracle that can simulate any given policy from any initial state with bounded loss --- a condition that is satisfied in the DRACC problem --- then the online learning problem can be solved with vanishing regret.
Our proof technique is based on a reduction to online learning with \emph{switching cost}, in which an online decision maker incurs an extra cost every time she switches from one arm to another.
We formally demonstrate this connection and further show how DRACC can be used in our proposed applications of stateful pricing.
\end{abstract}

%\end{abstract}

\section{Introduction}
\label{sec:introduction}
% Intro version 1.0
\label{section:introduction}
%why posted pricing?
Price posting is a common selling mechanism across various corners of e-commerce. Its applications span from more traditional domains such as selling flight tickets on Delta's website or selling products on Amazon, to more emerging domains such as selling cloud services on AWS or pricing ride-shares in Uber.
The prevalence of price posting comes from its several important advantages:
it is incentive compatible, simple to grasp, and can easily fit in an online (or dynamic) environment where buyers arrive sequentially over time.
Therefore, online posted pricing mechanisms, also known as dynamic pricing, have been studied quite extensively in computer science, operations research, and economics (for a comprehensive survey, see~\cite{den2015dynamic}).

%limitation of current online posted pricing literature is being statless
A very useful method for devising online posted prices is via \emph{vanishing-regret online learning} algorithms in an adversarial environment
\citep{bubeck2019multi, bubeck2017online, feldman2016online, blum2005near,blum2004online,  kleinberg2003value}.
Here, a sequence of buyers arrive, each associated with her own valuation function that is assumed to be devised by a malicious adversary, and the goal is to post a sequence of price vectors that perform almost as good as the best fixed pricing policy in hindsight.
Despite its success, a technical limitation of this method (shared by the aforementioned papers) forces the often less natural assumption of unlimited item supply to ensure that the selling platform is \emph{stateless}.
However, in many applications of online posted pricing, the platform is \emph{stateful};
indeed, prices can depend on previous sales that determine the platform's state.
Examples for such stateful platforms include selling resources of limited supply, in which the state encodes the number of remaining inventories of different products, and selling resources in cloud computing to schedule online jobs, in which the state encodes the currently scheduled jobs.

%limitation of stateful ones is being stochastic except Chawla et al. 
The above mentioned limitation is in sharp contrast to the posted prices literature that consider stochastic settings where the buyers' valuations are drawn independently and identically from unknown distributions
\citep{badanidiyuru2013bandits, babaioff2015dynamic, zhang2018occupation}, or independently from known distributions~\citep{chawla2010multi, feldman2014combinatorial, Chawla2017SST}.
By exploiting the randomness (and distributional knowledge) of the input and employing other algorithmic techniques, these papers cope with limited supply and occasionally, with more complicated stateful pricing scenarios.
However, the stochastic approach does not encompass the (realistic) scenarios in which the buyers' valuations are correlated in various complex ways, scenarios that are typically handled using adversarial models.
The only exception in this regard is the work of Chawla et al.~\cite{Chawla2017TRO} that takes a different approach:
they consider the online job scheduling problem, and given access to a collection of (truthful) posted price scheduling mechanisms, they show how to design a (truthful) vanishing-regret online scheduling mechanism against this collection in an adversarial environment.

Motivated by the abundance of stateful posted pricing platforms, and inspired by \cite{Chawla2017TRO}, we study the design of adversarial online learning algorithms with vanishing regret for a rather general online resource allocation framework.
In this framework, termed \emph{dynamic resource allocation with capacity constraints (DRACC)}, dynamic resources of limited inventories arrive and depart over time, and an online mechanism sequentially posts price vectors to (myopically) strategic buyers with adversarially chosen combinatorial valuations (refer to \Cref{sec:model} for the formal model).
The goal is to post a sequence of price vectors with the objective of maximizing revenue, while respecting the inventory restrictions of dynamic resources for the periods of time in which they are active.
We consider a full-information setting, in which the buyers' valuations are elicited by the platform after posting prices in each round of the online execution.

Given a collection of pricing policies for the DRACC framework, we aim to construct a sequence of price vectors that is guaranteed to admit a vanishing regret with respect to the best fixed pricing policy in hindsight.
Interestingly, our abstract framework is general enough to admit, as special cases, two important applications of stateful posted pricing, namely, online job-scheduling and matching over a dynamic bipartite graph;
these applications, for which existing online learning techniques fail to obtain vanishing regret, are discussed in detail in \Cref{sec:applications}.

\vspace{-2mm}
%%%%%%%%%%%%%%%%%%%%%%%%%%%%%%%%%%%%%%%

\paragraph{Our contributions and techniques.}
Our main result is a vanishing-regret posted price mechanism for the DRACC problem (refer to \Cref{sec:main-technical} for a formal exposition).
\vspace{-0.6mm}
\begin{displayquote}
\emph{For any DRACC instance with $T$ users and for any collection $\Gamma$ of
pricing policies, the regret of our proposed posted price mechanism (in terms of expected revenue) with respect to the in-hindsight best policy in $\Gamma$ is sublinear in $T$.}
\end{displayquote}
\vspace{-0.7mm}
%We further establish that our modeling assumptions in \Cref{section:dynamic-resource-allocation-capacity-constraints}, which are required to obtain the above result, are necessary, i.e., by lifting any of them no posted pricing mechanism can obtain regret smaller than $\Omega(T)$ (for details, see \Cref{lemma_impossibility_result_for_W}, \Cref{lemma_impossibility_result_for_C} and \Cref{lemma_fifo_is_essential} in \Cref{section:assumption-necessity}).

We prove this result by abstracting away the details of the pricing problem and considering a more general stateful decision making problem.
To this end, we introduce a new framework, termed \emph{dynamic deterministic Markov decision process (Dd-MDP)}, which generalizes the classic \emph{deterministic MDP} problem to an adversarial online learning dynamic setting.
In this framework, a decision maker picks a feasible action for the current state of the MDP, not knowing the state transitions and the rewards associated with each transition;
the state transition function and rewards are then revealed.
The goal of the decision maker is to pick a sequence of actions with the objective of maximizing her total reward.
In particular, we look at vanishing-regret online learning, where the decision maker is aiming at minimizing her regret, defined with respect to the in-hindsight best fixed policy (i.e., a mapping from states to actions) among the policies in a given collection $\Gamma$.

Not surprisingly, vanishing-regret online learning is impossible for this general problem (see \Cref{prop:no-noregret-general-DdMDP}).
To circumvent this difficulty, we introduce a structural condition on Dd-MDPs that enables online learning with vanishing regret.
This structural condition ensures the existence of an ongoing \emph{chasing oracle} that allows one to simulate a given fixed policy from any initial state, irrespective of the actual current state, while ensuring a small (vanishing) \emph{chasing regret}.
The crux of our technical contribution is cast in proving that the Dd-MDPs induced by DRACC instances satisfy this \emph{chasability} condition.

Subject to the chasability condition, we establish a reduction from designing vanishing-regret online algorithms for Dd-MDP to the extensively studied (classic stateless) setting of \emph{online learning with switching cost}~\citep{Kailai2005EAO}.
At high level, we have one arm for each policy in the given collection $\Gamma$ and employ the switching cost online algorithm to determine the next policy to pick.
Each time this algorithm suggests a switch to a new policy
$\gamma \in \Gamma$,
we invoke the chasing oracle that attempts to simulate $\gamma$, starting from the current state of the algorithm which may differ from $\gamma$'s current state.
In summary, we obtain the following result (see \Cref{theorem_regret_of_hola} for a formal exposition).

\begin{displayquote}
\emph{For any $T$-round Dd-MDP instance that satisfies the chasability condition and for any collection $\Gamma$ of policies, the regret of our online learning algorithm with respect to the in-hindsight best policy in $\Gamma$ is sublinear (and optimal) in $T$.}
\end{displayquote}

We further study the bandit version of the above problem, where the state transition function is revealed at the end of each round, but the learner only observes the current realized reward instead of the complete reward function.
By adapting the chasability condition to this setting, we obtain near optimal regret bounds.
See \Cref{theorem_regret_of_C_and_S_FLP} and \Cref{corollary-bndit} in \Cref{section:bandit} for a formal statement.

Our abstract frameworks, both for stateful decision making and stateful pricing, are rather general and we believe that they will turn out to capture many natural problems as special cases (on top of the applications discussed in \Cref{section:applications}).
\paragraph{Additional related work and discussion.}

%%%%%%%%%%%%%%%%%%%%%%%%%%%%%%%%%%%%%%%
In the DRACC problem, the class of feasible prices at each time $t$ is determined by the remaining inventories, which in turn depends on the prices picked at previous times $t' < t$. 
This kind of dependency cannot be handled by the conventional online learning algorithms, such as follow-the-perturbed-leader \cite{Kailai2005EAO} and EXP3 \cite{Auer2002NMB}. 
That is why we aim for the stateful model of online learning, which allows a certain degree of dependence on the past actions. 

Several attempts have been made to formalize and study stateful online learning models. 
The authors of \cite{Arora2012OBL, feldman2016online} consider an online learning framework where the reward (or cost) at each time depends on the $k$ recent actions for some fixed $k > 0$. 
This framework can be viewed as a reward function that depends on the system's state that, in this case, encodes the last $k - 1$ actions.

There is an extensive line of work on online learning models that address general multi-state systems, typically formalized by means of stochastic \cite{Eyal2005EMD, Guan2014MVL, Yu2008MDP, Abbasi2013OLM, Neu2014OMD} or deterministic \cite{Dekel2013BRA} MDPs. 
The disadvantage of these models from our perspective is that they all have at least one of the following two restrictions: 
(a)
all actions are always feasible regardless of the current state \cite{Abbasi2013OLM, Eyal2005EMD, Guan2014MVL, Yu2008MDP};
or
(b)
the state transition function is fixed (static) and known in advance \cite{Dekel2013BRA, Eyal2005EMD, Guan2014MVL, Neu2014OMD, Yu2008MDP}.

In the DRACC problem, however, not all actions (price vectors) are feasible for every state and the
state transition function at time $t$ is revealed only after the decision maker has committed to its action. 
Moreover, the aforementioned MDP-based models require a certain type of state \emph{connectivity} in the sense that the Markov chain induced by each action should be irreducible \cite{Abbasi2013OLM, Eyal2005EMD, Guan2014MVL, Neu2014OMD, Yu2008MDP} or at least the union of all induced Markov chains should form a strongly connected graph \cite{Dekel2013BRA}. 
In contrast, in the DRACC problem, depending on the inventories of the resources, it may be the case that certain {\Surplus} vectors can never be reached (regardless of the decision maker's actions).

On the algorithmic side, a common feature of all aforementioned online learning models is that for every instance, there exists some $k > 0$ that can be computed in a preprocessing stage (and does not depend on $T$) such that the online learning can ``catch'' the state (or distribution over states) of any given sequence of actions in \emph{exactly} $k$ time units. 
While this feature serves as a corner stone for the existing online learning algorithms, it is not present in our model, hence our online learning algorithm has to employ different ideas.

In \cite{Devanur2011NOO, Kesselheim2014PBD, Agrawal2014FAO}, a family of online resource allocation problems is investigated under a different setting from ours. The {\Items} in their problem models are static, which means that every {\Item} is revealed at the beginning, and remains active from the first {\Buyer} to the last one. Different from our adversarial model, these papers take different stochastic settings on the {\Buyers}, such as the random permutation setting where a fixed set of {\Buyers} arrive in a random order \cite{Kesselheim2014PBD, Agrawal2014FAO}, and the random generation setting where the parameters of each {\Buyer} are drawn from some distribution \cite{Devanur2011NOO, Agrawal2014FAO}. In these papers, the assignment of the resources to the requests are fully determined by a single decision maker, and the decision for each request depends on the revealed parameters of the current request and previous ones. By contrast, we study the scenario where each strategic {\Buyer} makes her own decision of choosing the resources, and the price posted to each {\Buyer} should be specified independently of the valuation of the current {\Buyer}.

\section{Model and Definitions}
\label{sec:model}

% (commented for NIPS) In this section, we formalize the problem of designing posted price mechanisms with vanishing regret for the DRACC setting, and then formalize its generalization to online learning with vanishing regret for the abstract Dd-MDP setting. In principle, both of these problems are examples of online learning in a stateful environment. We further dig into the connection between the two settings in \Cref{section:PPP-for-DRACC}.

\paragraph{The DRACC problem.}
Consider $N$ dynamic \emph{resources} and $T$ strategic myopic \emph{users} arriving sequentially over \emph{rounds}
$t = 1, \dots, T$,
where round $t$ lasts over the time interval
$[t, t + 1)$.
Resource
$i \in [N]$
arrives at the beginning of round $t_{a}(i)$ and departs at the end of round $t_{e}(i)$, where
$1 \leq t_{a}(i) \leq t_{e}(i) \leq T$;
upon arrival, it includes 
$\Capacity(i) \in \mathbb{Z}_{> 0}$
units.
We say that resource $i$ is \emph{active} at time $t$ if
$t_{a}(i) \leq t \leq t_{e}(i)$
and denote the set of resources active at time $t$ by
$A_t \subseteq [N]$.
Let $C$ and $W$ be upper bounds on
$\max_{i \in [N]} \Capacity(i)$
and
$\max_{t \in [T]} |A_{t}|$,
respectively.

The arriving user at time $t$ has a \emph{valuation function}
$\Valuation_{t} : 2^{\Active_{t}} \rightarrow [0, 1)$
that determines her value $\Valuation_{t}(A)$ for each subset
$A \subseteq \Active_{t}$
of resources active at time $t$.
We assume that
$\Valuation_{t}(\emptyset) = 0$
and that the users are quasi-linear, namely, if a subset $A$ of resources is allocated to user $t$ and she pays a total payment of $q$ in return, then her utility is
$\Valuation_t(A) - q$.
A family of valuation functions that receives a separated attention in this paper is that of \emph{$k_{t}$-demand} valuation functions, where user $t$ is associated with an integer parameter
$1 \leq k_{t} \leq |A_{t}|$
and with a value
$w^{i}_{t} \in [0, 1)$
for each active resource
$i \in A_{t}$
so that her value for a subset
$A \subseteq A_{t}$
is
$\max_{A' \subseteq A : \lvert A' \rvert \leq k_{t}} \sum_{i \in A'} w^{i}_{t}$.
%$\underset{A' \subseteq A : \lvert A' \rvert \leq k_{t}}{\max}
%\sum_{i \in A'} w^{i}_{t}$.

%%%%%%%%%%%%%%%%%%%%%%%%%%%%%%%%%%%%%%%
\paragraph{Stateful posted price mechanisms.}
We restrict our attention to dynamic posted price mechanisms that work based on the following protocol.
In each round
$t \in [T]$,
the mechanism first realizes which resources
$i \in [N]$
arrive at the beginning of round $t$, together with their initial capacity $\Capacity(i)$, and which resources departed at the end of round
$t - 1$,
thus updating its knowledge of $A_{t}$.
It then posts a \emph{price vector}
$\PriceVec_{t} \in (0, 1]^{\Active_{t}}$
that determines the price $\PriceVec_{t}(i)$ of each resource
$i \in \Active_{t}$
at time $t$.
Following that, the mechanism elicits the valuation function $\Valuation_t$ of the current user $t$ and allocates (or in other words sells) one unit of each resource in the demand set $\Selected_{t}^{\PriceVec_{t}}$ to user $t$ at a total price of $\hat{\Pay}_{t}^{\PriceVec_{t}}$, where
\begin{equation}\textstyle
\Selected_{t}^{\PriceVec}
\, = \,
\argmax_{A \subseteq \Active_t}
\left\{ \Valuation_t(A) - \sum_{i \in A} \PriceVec(i) \right\}
\qquad \text{and} \qquad
\hat{\Pay}_{t}^{\PriceVec}
\, = \,
\sum_{i\in \Selected_{t}^{\PriceVec}}\PriceVec(i) \label{formula_def_demand_set_and_payment}
\end{equation}
for any price vector
$\PriceVec \in (0, 1]^{\Active_{t}}$,
consistently breaking $\argmax$ ties according to the lexicographic order on $A_{t}$.
A virtue of posted price mechanisms is that if the choice of $\PriceVec_{t}$ does not depend on $\Valuation_{t}$, then it is dominant strategy for (myopic) user $t$ to report her valuation $\Valuation_{t}$ truthfully.

Let
$\SurplusVec_{t} \in \{ 0, 1, \dots, C \}^{\Active_t}$
be the \emph{inventory} vector that encodes the number $\SurplusVec_{t}(i)$ of units remaining from resource
$i \in \Active_{t}$
at time
$t = 1, \dots, T$.
Formally, if
$t_{a}(i) = t$,
then
$\SurplusVec_{t}(i) = \Capacity(i)$;
and if (a unit of) $i$ is allocated to user $t$ and $i$ is still active at time
$t + 1$,
then
$\SurplusVec_{t + 1}(i) = \SurplusVec_{t}(i) - 1$.
We say that a price vector $\PriceVec$ is \emph{feasible} for the inventory vector $\SurplusVec_{t}$ if
$\PriceVec(i) = 1$
for every
$i \in \Active_{t}$
such that
$\SurplusVec_{t}(i) = 0$,
that is, for every (active) resource $i$ exhausted by round $t$.
To ensure that the resource inventory is not exceeded, we require that the posted price vector $\PriceVec_{t}$ is feasible for $\SurplusVec_{t}$ for every
$1 \leq t \leq T$;
indeed, since $\Valuation_{t}$ is always strictly smaller than $1$, this requirement ensures that the utility of user $t$ from any resource subset
$A \subseteq \Active_{t}$
that includes an exhausted resource is negative, thus preventing $A$ from becoming the selected demand set, recalling that the utility obtained by user $t$ from the empty set is $0$.

%The crux of posted price mechanisms is that if the price vector
%$\PriceVec$ is determined independently of the valuation function
%$\Valuation_{t}$, then it is a dominant strategy for user $t$ to report
%$\Valuation_{t}$ truthfully.
%Following this (classic) observation, we define the DRACC problem in terms of
%an online \emph{decision maker} that at time
%$t = 1, \dots, T$: \\
%(1) learns of the resources
%$i \in [N]$
%whose arrival time is
%$t_{a}(i) = t$
%and of their capacities $\Capacity(i)$; \\
%(2) posts a price vector $\PriceVec_{t}$ for $\Active_{t}$; \\
%(3) receives the valuation function $\Valuation_{t}$ of user $t$;
%and \\
%(4) learns of the resources
%$i \in [N]$
%whose expiry time is
%$t_{e}(i) = t$. \\
%
%We emphasize that the decision maker does not hold any information about the
%future and that her decisions (i.e., the posted price vectors) are
%irrevocable.

In this paper, we aim for posted price mechanisms whose objective is to maximize the extracted \emph{revenue} defined to be the total expected payment
$\Ex[\sum_{t = 1}^{T} \hat{\Pay}_{t}^{\PriceVec_{t}}]$
received from all users, where the expectation is over the mechanism's internal randomness.\footnote{The techniques we use in this paper are applicable also to the objective of maximizing the social welfare.}

%%%%%%%%%%%%%%%%%%%%%%%%%%%%%%%%%%%%%%%
\paragraph{Adversarial online learning over pricing policies.}
%\& Regret.}
%%%%%%%%%%%%%%%%%%%%%%%%%%%%%%%%%%%%%%%
To measure the quality of the aforementioned posted price mechanisms, we consider an adversarial online learning framework, where at each time
$t \in [T]$,
the decision maker picks the price vector $\PriceVec_t$ and an adaptive adversary simultaneously picks the valuation function $\Valuation_t$.
The resource arrival times $t_{a}(i)$, departure times $t_{e}(i)$, and initial capacities $\Capacity(i)$ are also determined by the adversary.
We consider the full information setting, where the valuation function $\Valuation_{t}$ of user $t$ is reported to the decision maker at the end of each round $t$.
It is also assumed that the decision maker knows the parameters $C$ and $W$ upfront and that these parameters are independent of the instance length $T$.

A (feasible) \emph{pricing policy} $\PricePol$ is a function that maps each inventory vector
$\SurplusVec \in \{ 0, 1, \dots, C\}^{A_{t}}$,
$t \in [T]$,
to a price vector
$\PriceVec = \PricePol(\SurplusVec)$,
subject to the constraint that $\PriceVec$ is feasible for $\SurplusVec$.\footnote{%
The seemingly more general setup, where the time $t$ is passed as an argument to $\PricePol$ on top of $\SurplusVec$, can be easily reduced to our setup (e.g., by introducing a dummy resource $i_{t}$ active only in round $t$).}
The pricing policies are used as the benchmarks of our online learning framework:
Given a pricing policy $\PricePol$, consider a decision maker that repeatedly plays according to $\PricePol$;
namely, she posts the price vector
$\PriceVec^{\PricePol}_{t} = \PricePol(\SurplusVec^{\PricePol}_{t})$
at time
$t = 1, \dots, T$,
where $\SurplusVec^{\PricePol}_{t}$ is the inventory vector at time $t$ obtained by applying $\PricePol$ recursively on previous inventory vectors $\SurplusVec^{\PricePol}_{t'}$ and posting prices
$\PricePol(\SurplusVec^{\PricePol}_{t'})$
at times
$t' = 1, \dots, t - 1$.
Denoting
$\hat{\Pay}_{t}^{\PricePol} = \hat{\Pay}_{t}^{\PriceVec^{\PricePol}_{t}}$,
the revenue of this decision maker is given by
$\sum_{t = 1}^{T} \hat{\Pay}_{t}^{\PricePol}$.

%A (feasible) \emph{pricing policy} $\PricePol$ is a sequence
%$\PricePol = \{ \PriceVec_{t}^{\PricePol} \}_{t = 1}^{T}$
%of price vectors so that
%$\PriceVec_{t}^{\PricePol}$
%is a feasible price vector for
%$\SurplusVec_{t}^{\PricePol}$,
%where $\SurplusVec_{t}^{\PricePol}$ is the inventory vector at time $t$ obtained by applying $\PricePol$ recursively on previous inventory vectors $\SurplusVec_{t'}^{\PricePol}$ and posting prices
%$\PriceVec_{t'}^{\PricePol}$
%at times
%$t' = 1, \ldots, t - 1$.

Now, consider a collection $\Gamma$ of pricing policies.
The quality of a posted price mechanism
$\{ \PriceVec_{t} \}_{t = 1}^{T}$
is measured by means of the decision
maker's \emph{regret} that compares her own revenue to the revenue generated
by the in-hindsight best pricing policy in $\Gamma$.
Formally, the regret (with respect to $\Gamma$) is defined to be
\[\textstyle
\max_{\PricePol \in \PricePols}
\sum_{t = 1}^{T} \hat{\Pay}_{t}^{\PricePol}
-
\Ex \left[
\sum_{t = 1}^{T} \hat{\Pay}_{t}^{\PriceVec_{t}}
\right]
\, ,
\]
where the expectation is taken over the decision maker's randomness.
The mechanism is said to have \emph{vanishing regret} if it is guaranteed that the decision maker's regret is sublinear in $T$, which means that the average regret per
time unit vanishes as $T \rightarrow \infty$.
%The mechanism is said to be \emph{no-regret} if it is guaranteed that the decision maker's regret is sublinear in $T$, which means that the average regret per
%time unit vanishes as $T \rightarrow \infty$.

% \paragraph{Modeling Assumptions.} In our treatment of DRACC, we make two assumptions:
% \begin{itemize}
% \item \emph{Bounded parameters}: The collection $\Gamma$ of pricing policies is given as input by an adversary, and its size is assumed to satisfy
% \begin{equation}
% C\cdot W \cdot \log |\Gamma| = o(T) \, . \label{formula_DRACC_Bounded_Parameters}
% \end{equation}
% \item \emph{Finite discrete prices:} posted prices are restricted to a finite subset
% $\mathcal{P} \subset (0, 1]$ that explicitly contains the price $1$.\footnote{Since the collection $\Gamma$ is specified by the adversary, its size does not directly depend on $|\mathcal{P}|$, $C$ and $W$. Thus, we do not make any further assumption on how $|\mathcal{P}|$, $C$ and $W$ are bounded.}
% \end{itemize}

\section{Dynamic Posted Pricing via Dd-MDP with Chasability}
\label{sec:main-technical}
\label{section:HOLA}
The online learning framework underlying the DRACC problem as defined in \cref{sec:model} is stateful with the inventory vector $\SurplusVec$ playing the role of the framework's state.
In the current section, we first introduce a generalization of this online learning framework in the form of a stateful online decision making, formalized by means of \emph{dynamic deterministic Markov decision processes (Dd-MDPs)}.
Following that, we propose a structural condition called \emph{chasability} and show that under this condition, the Dd-MDP problem is amenable to vanishing-regret online learning algorithms.
This last result is obtained through a reduction to the extensively studied problem of ``experts with switching cost'' ~\cite{Kailai2005EAO}.
Finally, we prove that the Dd-MDP instances that correspond to the DRACC problem indeed satisfy the chasability condition.

\subsection{Viewing DRACC as a Dd-MDP}
\label{sec:Dd-MDP-prelim}

A (static) \emph{deterministic Markov decision process (d-MDP)} is defined over a set $\states$ of states and a set $\actions$ of actions.
Each state
$s \in \states$
is associated with a subset
$X_{s} \subseteq \actions$
of actions called the feasible actions of $s$.
A \emph{state transition function} $g$ maps each state
$s \in \states$
and action
$x \in X_{s}$
to a state
$g(s, x) \in \states$.
This induces a directed graph over $\states$, termed the \emph{state transition graph}, where an edge labeled by
$\langle s, x \rangle$
leads from node $s$ to node $s'$ if and only if
$g(s, x) = s'$.
The d-MDP also includes a \emph{reward function} $f$ that maps each state-action pair $\langle s, x \rangle$ with $s \in \states$
and 
$x \in X_{s}$ to a real value in $[0, 1]$.

\paragraph{Dynamic deterministic MDPs.}
Notably, static d-MDPs are \emph{not} rich enough to capture the dynamic aspects of the DRACC problem.
We therefore introduce a more general object where the state transition and reward functions are allowed to develop in an (adversarial) dynamic fashion.

Consider a sequential game played between an online decision maker and an adversary.
As in static d-MDPs, the game is defined over a set $\states$ of states, a set $\actions$ of actions, and a feasible action set $X_s$ for each
$s \in \states$.
We further assume that the state and action sets are finite.
The game is played in
$T \in \mathbb{N}$
rounds as follows.
The decision maker starts from an initial state
$s_{1} \in \states$. In each round
$t = 1, \ldots, T$, she
plays a (randomized) feasible action
$x_t \in X_{s_{t}}$,
where
$s_{t} \in \states$
is the state at the beginning of round
$t$.
Simultaneously, the adversary selects the state transition function $g_{t}$ and the reward function $f_{t}$.
The decision maker then moves to a new state
$s_{t + 1} = g_{t}(s_{t}, x_{t})$
(which is viewed as a movement along edge
$\langle s_{t}, x_{t} \rangle$
in the state transition graph induced by $g_{t}$),
obtains a reward
$f_{t}(s_{t}, x_{t})$,
and finally, observes $g_{t}$ and $f_{t}$ as the current round's (full information) feedback.\footnote{%
No (time-wise) connectivity assumptions are made for the dynamic transition graph induced by
$\{ g_{t} \}_{t = 1}^{T}$,
hence it may not be possible to devise a path between two given states as is done in \cite{Dekel2013BRA} for static d-MDPs.}
The game then advances to the next round
$t + 1$.
The goal is to maximize the expected total reward
$\mathbb{E} [ \sum_{t \in [T]} f_{t}(s_{t}, x_{t})]$.

%Usually, the transition digraph is assumed to be strongly connected.

\paragraph{Policies, simulation, \& regret.}
A (feasible) \emph{policy}
$\gamma: \states \mapsto \actions$
is a function that maps each state
$s \in \states$
to an action
$\gamma(s) \in X_{s}$.
A \emph{simulation} of policy $\gamma$ over the round interval
$[1, T]$
is given by the state sequence
$\{ \FPS{t}{\gamma} \}_{t = 1}^{T}$
and the action sequence
$\{ \FPA{t}{\gamma} \}_{t = 1}^{T}$
defined by setting
\begin{equation}
\FPS{t}{\gamma} \triangleq
\begin{cases}
s_{1} & \text{if } t = 1 \\
g_{t - 1} \left( \FPS{t - 1}{\gamma}, \FPA{t - 1}{\gamma} \right) & \text{if } t > 1
\end{cases}
\quad \text{and} \quad
\FPA{t}{\gamma} \triangleq
\begin{cases}
\gamma(s_{1}) & \text{if } t = 1 \\
\gamma \left( \FPS{t}{\gamma} \right) & \text{if } t > 1
\end{cases} \, .
\label{formula_fixed_policy_states_and_actions}
\end{equation}
The cumulative reward obtained by this simulation of $\gamma$ is given by
$\sum_{t \in [T]} f_{t} \left( \FPS{t}{\gamma}, \FPA{t}{\gamma} \right)$.

Consider a decision maker that plays the sequential game by following the (randomized) state sequence
$\{ s_{t} \}_{t = 1}^{T}$
and action sequence
$\{ x_{t} \}_{t = 1}^{T}$,
where
$x_{t} \in X_{s_{t}}$
for every
$1 \leq t \leq T$.
For a (finite) set $\Gamma$ of policies, the decision maker's \emph{regret} with respect to $\Gamma$ is defined to be
\begin{equation}\textstyle
\max_{\gamma \in \Gamma}
\sum_{t \in [T]} f_{t} \left( \FPS{t}{\gamma}, \FPA{t}{\gamma} \right)
\, - \,
\sum_{t \in [T]} \mathbb{E} \left[ f_{t}(s_{t}, x_{t}) \right] \, . \label{formula_policy_regret}
\end{equation}
 
%An online decision maker is said to be \emph{no-regret} if the upper bound on the regret is sublinear in $T$, which means that the average regret vanishes as $T \rightarrow \infty$. Our objective is to develop no-regret online decision making algorithms for this setting. 

\subsubsection*{\emph{Relation to the DRACC Problem}}
\vspace{-1mm}
Dynamic posted pricing for the DRACC problem can be modeled as a Dd-MDP.
To this end, we identify the state set $\states$ with the set of possible inventory vectors $\SurplusVec_{t}$,
$t = 1, \dots, T$.
If state
$s \in \states$
is identified with inventory vector $\SurplusVec_{t}$, then we identify $X_{s}$ with the set of price vectors feasible for $\SurplusVec_{t}$.
The reward function $f_{t}$
is defined by setting
\begin{equation}
f_{t}(s, x)
\, = \,
\ChoicePayment{t}{x}  \, , \label{formula_DRACC_reward_function}
\end{equation}
where $\ChoicePayment{t}{x}$ is defined as in Eq.~\eqref{formula_def_demand_set_and_payment}, recalling that the valuation function $v_{t}$, required for the computation of $\ChoicePayment{t}{x}$, is available to the decision maker at the end of round $t$.
%
%Due to Eq.~\eqref{formula_def_demand_set_and_payment}, computing the reward function relies on $v_{t}$. In the full information setting, $v_t$ is available upon the arrival of the {\Buyer} $t$ due to the truthfulness of our pricing mechanism. In the bandit setting $v_t$ is not observed, but again because of truthfulness the correct realization of this reward function for state-action pair $(\SurplusVec_t,\PriceVec_t)$ is observed.
%
As for the state transition function $g_{t}$, the new state
$s' = g_{t}(s, x)$
is the inventory vector obtained by posting the price vector $x$ to user $t$ given the inventory vector $s$, namely,
\begin{equation*}
s'(i)
\, = \,
\begin{cases}
s(i) - 1_{i \in \FAVitems{t}{x}} & \text{if } i \in A_{t + 1} \cap A_{t} \\
\MultiSymB{i} & \text{if } i \in A_{t + 1} \setminus A_{t}
\end{cases} \, .
\end{equation*}
Given the aforementioned definitions, the notion of (pricing) policies and their recursive simulations and the notion of regret translate directly from the DRACC setting to that of Dd-MDPs.

%called ``\emph{Chaseability}''. %Under this condition, roughly speaking, there exist states termed as \emph{dominant states} that if a policy can magically switch to them at any time, they always admit a larger future accumulated reward. Moreover, there  exists a switching procedure, termed as \emph{chasing oracle}, that can switch the current state of any online algorithm to a dominant state with negligible loss in reward. Interestingly, with this condition the Dd-MDP admits no-regret policies. In this section, we show how to design such policies with provably optimal regret bounds. 
\subsection{The Chasability Condition}
\label{sec:dominance}
As the Dd-MDP framework is very inclusive, it is not surprising that in general, it does not allow for vanishing regret.

\begin{proposition}
\label{prop:no-noregret-general-DdMDP}
For every online learning algorithm, there exists a $T$-round Dd-MDP instance for which the algorithm's regret is
$\Omega (T)$.
\end{proposition}

\begin{proof}
Consider a simple scenario where there are only two states
$\lbrace s, s' \rbrace$
with $s$ being the initial state and two actions
$\lbrace x, x' \rbrace$
that are feasible for both states.
Without loss of generality, let $x$ be the action that the decision maker's algorithm chooses with probability at least
$1 / 2$
at time
$t = 1$.
Now, consider an adversary that works in the following manner:
It sets
$f_{t}(s, \cdot) = 1$
and
$f_{t}(s', \cdot) = 0$
for every
$t \in [T]$.
Regarding the state transition, the adversary sets
$g_{1}(s, x) = s'$
and
$g_{1}(s, x') = s$;
and for every
$t \in [2, T]$,
it sets
$g_{t}(s, \cdot) = s$
and
$g_{t}(s', \cdot) = s'$.
In such case, the expected cumulative reward of the decision maker is at most
$1 + T / 2$,
while the policy that always plays action $x_1$ obtains a cumulative reward of $T$.
\end{proof}

As a remedy to the impossibility result established in \cref{prop:no-noregret-general-DdMDP}, we introduce a structural condition for Dd-MDPs that makes them amenable to online learning with vanishing regret.

%%%%%%%%%%%%%%%%%%%%%%%%%%%%%%%%%%%%%%%%%%%%%%%%%%%%%%%%%%%%%%%%%%%%%%%%%%%%%%

\begin{definition}[\emph{Chasability condition for Dd-MDPs}]
\label{def:chase}
A Dd-MDP instance is called \emph{$\sigma$-chasable} for some
$\sigma > 0$
if it admits an ongoing \emph{chasing oracle} $\mathcal{O}^{\mathtt{Chasing}}$ that works as follows for any given target policy 
$\gamma \in \Gamma$.
The chasing oracle is invoked at the beginning of some round $\tinit$ and provided with an initial state
$\sinit \in \states$;
this invocation is halted at the end of some round
$\tfinal \geq \tinit$.
In each round
$\tinit \leq t \leq \tfinal$,
the chasing oracle generates a (random) action $\hat{x}(t)$ that is feasible for state
\begin{equation}\textstyle
\hat{s}(t)
\, = \,
\begin{cases}
\sinit & \text{if } t = \tinit \\
g_{t - 1} \left( \hat{s}(t - 1), \hat{x}(t - 1) \right) & \text{if } \tinit < t \leq \tfinal
\end{cases} \, ;
\label{formula_states_visited_in_local_online_learning}
\end{equation}
following that, the chasing oracle is provided with the Dd-MDP's state transition function
$f_{t}(\cdot, \cdot)$
and reward function
$g_{t}(\cdot, \cdot)$.
The main guarantee of $\mathcal{O}^{\mathtt{Chasing}}$ is that its \emph{chasing regret (CR)} satisfies
\[\textstyle
\textrm{CR}
\, \triangleq \,
\sum_{t = \tinit}^{\tfinal}
f_{t} \left( \FPS{t}{{\gamma}}, \FPA{t}{{\gamma}} \right)
-
\sum_{t = \tinit}^{\tfinal}
\mathbb{E} \left[ f_{t}\Big(\hat{s}(t), \hat{x}(t) \Big) \right]
\, \leq \,
\sigma \, .
\]
We emphasize that the initial state $\sinit$ provided to the chasing oracle may differ from $\FPS{\tinit}{{\gamma}}$.
\end{definition}

% See \Cref{sec:stateless_chaseability} for an adapted version of chasability in the bandit setting.

\subsubsection*{\emph{Relation to the DRACC Problem (continuted)}}
\vspace{-1mm}
Interestingly, the Dd-MDPs corresponding to DRACC instances are $\sigma$-chasable for
$\sigma = o (T)$,
where the exact bound on $\sigma$ depends on whether we consider general or $k_{t}$-demand valuation functions.
Before establishing these bounds, we show that the chasing oracle must be randomized.

\begin{proposition} \label{prop:deterministic-chasing-oracle}
There exists a family of $T$-round DRACC instances whose corresponding Dd-MDPs do not admit a deterministic chasing oracle with $o (T)$ chasing regret CR.
\end{proposition}

\begin{proof}
Consider an ongoing chasing oracle that is implemented in a deterministic manner for a DRACC instance with $\MaxMultiSym = 1$, $W = 2$. The adversary chooses initial step $\tinit$ and initial state $\sinit$ so that $\tinit = o(T)$, $|A_{\tinit}| = 2$, and $\VEC{\SurplusSym}_{\tinit} = \langle 0, 1 \rangle$. Note that throughout this proof, the {\Surplus} vectors and price vectors containing two elements are presented in an ordered way, which means that the first element corresponds to the {\Item} with the smaller index.

The target policy $\gamma$ is chosen to have $\FPremain{\tinit}{\gamma} = \langle 1, 1 \rangle$. Moreover, it maps every {\Surplus} vector to a price vector of $\langle \frac{2}{3}, \frac{1}{3} \rangle$. The adversary ensures the feasibility of such a policy by setting $t_{e}(i) = t$ for each {\Item} that is sold out at $t$ with the price vector generated by $\gamma$, and setting $t_{a}(i') = t + 1$ for a new {\Item}. With this setting, it holds for every $t \geq \tinit$ that $\FPremain{t}{\gamma} = \langle 1, 1 \rangle$.

The adversary configures the valuation functions $v_{t}$ for each $t \geq \tinit$ in an adaptive way, and ensures that for all such $t$
\begin{equation}
    \hat{\VEC{\SurplusSym}}_{t} = \langle 0, 1 \rangle \, . \label{formula_invariant_surplus_vector_for_deterministic_oracle}
\end{equation}
With the initial state $\sinit$ chosen by the adversary, Eq.~\eqref{formula_invariant_surplus_vector_for_deterministic_oracle} holds for $\tinit$. Suppose it holds for some $t \geq \tinit$. Then the price vector $\hat{\PriceVec}$ generated by the oracle must be in the form of $\langle 1, p \rangle$ for some $p \in (0, 1]$. Let $i$ and $i'$ be the two {\Items} in $A_{t}$ with $i < i'$. If $p \leq \frac{1}{3}$, the adversary sets $v_{t}(i) = \frac{2}{3}$ and $v_{t}(i') = \frac{1}{3}$. Then with the price vector generated by $\gamma$, payment $\frac{2}{3}$ is obtained from the {\Buyer} for {\Item} $i$, while the oracle obtains payment $\frac{1}{3}$ from the {\Buyer} for $i'$. The difference in rewards is
\begin{equation}
    f_{t}(\FPremain{t}{\gamma}, \FPprice{t}{\gamma}) - f_{t}(\hat{\VEC{\SurplusSym}}_{t}, \hat{\PriceVec}_{t}) = \frac{1}{3} \, . \label{formula_difference_in_payments_for_deterministic_oracle}
\end{equation}
Moreover, since $i$ is sold out with $\FPprice{t}{\gamma}$, the adversary sets $t_{e}(i) = t$ and $t_{a}(i'') = t + 1$ for a new {\Item} $i'' > i'$. In such case, it is guaranteed that Eq.~\eqref{formula_invariant_surplus_vector_for_deterministic_oracle} holds for $t + 1$.

For the case where $p > \frac{1}{3}$, it can be verified that Eq.~\eqref{formula_difference_in_payments_for_deterministic_oracle} still holds for $t$ and Eq.~\eqref{formula_invariant_surplus_vector_for_deterministic_oracle}  holds for $t + 1$ when the adversary sets $v_{t}(i) = v_{t}(i') = \frac{1}{3}$. Since Eq.~\eqref{formula_difference_in_payments_for_deterministic_oracle} is established for every $t \geq \tinit$, $CR = \frac{1}{3}(\tfinal - \tinit)$. With $\tinit = o(T)$, taking $\tfinal = T$ gives the desired bound.
\end{proof}

We now turn to study chasing oracles for DRACC instances implemented by randomized procedures.

\begin{theorem}\label{theorem_chasing_condition}
The Dd-MDPs corresponding to $T$-round DRACC instances with $k_{t}$-demand valuation functions are
$O (\sqrt{C W \cdot T})$-chasable.
\end{theorem}
\begin{proof}
Consider some DRACC instance and fix the target pricing policy
$\gamma \in \Gamma$;
in what follows, we identify $\gamma$ with a decision maker that repeatedly plays according to $\gamma$.
Given an initial round $\tinit$ and an initial inventory vector
$\hat{\VEC{\SurplusSym}}_{\tinit}$,
we construct a randomized chasing oracle $\mathcal{O}^{\mathtt{Chasing}}$ that works as follows until it is halted at the end of round
$\tfinal \geq \tinit$.
For each round
$\tinit \leq t \leq \tfinal$,
recall that $\SurplusVec^{\gamma}_{t}$ is the inventory vector at time $t$ obtained by running $\gamma$ from round $1$ to $t$, and let $\hat{\SurplusVec}_{t}$ be the inventory vector at time $t$ obtained by $\mathcal{O}^{\mathtt{Chasing}}$ 
as defined in Eq.~\eqref{formula_states_visited_in_local_online_learning}.
We partition the set $A_{t}$ of resources active at time $t$ into
$\mathtt{Good}_{t} =
\{ i \in A_{t} \mid \SurplusVec^{\gamma}_{t}(i) \leq \hat{\SurplusVec}_{t}(i) \}$
and
$\mathtt{Bad}_{t} = A_{t} \setminus \mathtt{Good}_{t}$.
%\[
%\mathtt{Bad}_{t}
%\, = \,
%\left\{ i \in A_{t} \mid \SurplusVec^{\gamma}_{t}(i) > \hat{\SurplusVec}_{t}(i) = 0 \right\}
%\qquad \text{and} \qquad
%\mathtt{Good}_{t}
%\, = \,
%A_{t} - \mathtt{Bad}_{t} \, .
%\]
In each round
$\tinit \leq t \leq \tfinal$,
the chasing oracle posts the ($|A_{t}|$-dimensional) all-$1$ price vector with probability $\epsilon$, where
$\epsilon \in (0, 1)$
is a parameter to be determined later on;
and it posts the price vector
\[
\hat{\PriceVec}_{t}
\, = \,
\begin{cases}
\PriceVec^{\gamma}_{t}(i) & \text{if } i \in \mathtt{Good}_{t} \\
1 & \text{if } i \in \mathtt{Bad}_{t}
\end{cases}
\]
with probability
$1 - \epsilon$,
observing that this price vector is feasible for $\hat{\SurplusVec}_{t}$ by the definition of $\mathtt{Good}_{t}$ and $\mathtt{Bad}_{t}$.
Notice that $\mathcal{O}^{\mathtt{Chasing}}$ never sells a resource
$i \in \mathtt{Bad}_{t}$
and that
$\hat{\PriceVec}_{t}(i) \geq \PriceVec^{\gamma}_{t}(i)$
for all
$i \in A_{t}$.
Moreover, if resource $i$ arrives at time
$t_{a}(i) = t > \tinit$,
then
$i \in \mathtt{Good}_{t}$.

To analyze the CR, we classify the rounds in
$[\tinit, \tfinal]$
into two classes called $\mathtt{Following}$ and $\mathtt{Missing}$:
round $t$ is said to be $\mathtt{Missing}$ if at least one (unit of a) resource in $\mathtt{Bad}_{t}$ is sold by $\gamma$ in this round;
otherwise, round $t$ is said to be $\mathtt{Following}$.
For each $\mathtt{Following}$ round $t$, if $\mathcal{O}^{\mathtt{Chasing}}$ posts $\hat{\PriceVec}_{t}$ in round $t$, then $\mathcal{O}^{\mathtt{Chasing}}$ sells exactly the same resources as $\gamma$ for the exact same prices;
otherwise ($\mathcal{O}^{\mathtt{Chasing}}$ posts the all-$1$ price vector in round $t$), $\mathcal{O}^{\mathtt{Chasing}}$ does not sell any resource.
Hence, the CR increases in round $t$ by at most $\epsilon$ in expectation.
For each $\mathtt{Missing}$ round $t$, the CR increases in round $t$ by at most $1$.
Therefore the total CR over the interval
$[\tinit, \tfinal]$
is upper bounded by
$\epsilon \cdot \mathbb{E}[\#\mathtt{F}] + \mathbb{E}[\#\mathtt{M}]
\leq
\epsilon \cdot T + \mathbb{E}[\#\mathtt{M}]$,
where $\#\mathtt{F}$ and $\#\mathtt{M}$ denote the number of $\mathtt{Following}$ and $\mathtt{Missing}$ rounds, respectively.

To bound $\mathbb{E}[\#\mathtt{M}]$, we introduce a potential function $\phi(t)$,
$\tinit \leq t \leq \tfinal$,
defined by setting
\[\textstyle
\phi(t)
\, = \,
\sum_{i \in \mathtt{Bad}_{t}} \SurplusVec^{\gamma}_{t}(i) - \hat{\SurplusVec}_{t}(i)
\]
By definition,
$\phi(\tinit) \leq C W$
and
$\phi(\tfinal) \geq 0$.
We argue that $\phi(t)$ is non-increasing in $t$.
To this end, notice that if $t$ is a $\mathtt{Following}$ round, then
$\mathtt{Bad}_{t + 1} \subseteq \mathtt{Bad}_{t}$,
hence
$\phi(t + 1) \leq \phi(t)$.
If $t$ is a $\mathtt{Missing}$ round and $\mathcal{O}^{\mathtt{Chasing}}$ posts the all-$1$ price vector, then
$\phi(t + 1) < \phi(t)$
as $\mathcal{O}^{\mathtt{Chasing}}$ sells no resource whereas $\gamma$ sells at least one (unit of a) resource in $\mathtt{Bad}_{t}$.
So, it remains to consider a $\mathtt{Missing}$ round $t$ in which $\mathcal{O}^{\mathtt{Chasing}}$ posts the price vector $\hat{\PriceVec}_{t}$.
Let $S^{\gamma}$ and $\hat{S}$ be the sets of (active) resources sold by $\gamma$ and $\mathcal{O}^{\mathtt{Chasing}}$, respectively, in round $t$ and notice that a resource
$i \in \hat{S} \setminus S^{\gamma}$
may move from
$i \in \mathtt{Good}_{t}$
to
$i \in \mathtt{Bad}_{t + 1}$.
The key observation now is that since $v_{t}$ is a $k_{t}$-demand valuation function, it follows that
$S^{\gamma} \cap \mathtt{Good}_{t} \subseteq \hat{S} \cap \mathtt{Good}_{t}$,
thus
$|S^{\gamma} \cap \mathtt{Bad}_{t}| \geq |\hat{S} \setminus S^{\gamma}|$.
As both $\gamma$ and $\mathcal{O}^{\mathtt{Chasing}}$ sell exactly one unit of each resource in $S^{\gamma}$ and $\hat{S}$, respectively, we conclude that
$\phi(t + 1) \leq \phi(t)$.

Therefore, $\mathbb{E}[\#\mathtt{M}]$ is upper bounded by
$C W$
plus the expected number of $\mathtt{Missing}$ rounds in which $\phi(t)$ does not decrease.
Since $\phi(t)$ strictly decreases in each $\mathtt{Missing}$ round $t$ in which $\mathcal{O}^{\mathtt{Chasing}}$ posts the all-$1$ price vector, it follows that the number of $\mathtt{Missing}$ rounds in which $\phi(t)$ does not decrease is stochastically dominated by a negative binomial random variable $Z$ with parameters $C W$ and $\epsilon$.
Recalling that
$\mathbb{E}[Z] = (1 - \epsilon) \cdot C W / \epsilon$,
we conclude that
$\mathbb{E}[\#\mathtt{M}]
\leq
C W + \mathbb{E}[Z]
=
C W / \epsilon$.
The assertion is now established by setting
$\epsilon = \sqrt{C W / T}$.
\end{proof}

\begin{remark} \label{remark:generalize-k-demand}
\Cref{theorem_chasing_condition} can be in fact extended -- using the exact same line of arguments -- to a more general family of valuation functions $v_{t}$ defined as follows.
Let $\PriceVec$ be a price vector, 
$B \subseteq A_{t}$
be a subset of the active resources, and $\PriceVec'$ be the price vector obtained from $\PriceVec$ by setting
$\PriceVec'(i) = 1$
if
$i \in B$;
and
$\PriceVec'(i) = \PriceVec(i)$
otherwise.
Then,
$|\hat{A}^{\PriceVec}_{t} \cap B| \geq
|\hat{A}^{\PriceVec'}_{t} \setminus \hat{A}^{\PriceVec}_{t}|$.
Besides $k_t$-demand valuations, this class of valuation functions includes OXS valuations~\cite{lehmann2006combinatorial} and single-minded valuations~\cite{lehmann2002truth}.
\end{remark}
\vspace{1mm}
\begin{theorem}\label{theorem_chasing_condition-general}
The Dd-MDPs corresponding to $T$-round DRACC instances with arbitrary valuation functions are
$O\left(T^{\tfrac{C W}{C W + 1}}\right)$-chasable.
\end{theorem}
\vspace{-4mm}
\begin{proof}
The proof follows the same line of arguments as that of \Cref{theorem_chasing_condition}, only that now, it no longer holds that the potential function $\phi(t)$ is non-increasing in $t$.
However, it is still true that
(I)
$0 \leq \phi(t) \leq C W$
for every
$\tinit \leq t \leq \tfinal$;
(II)
if
$\tinit \leq t < \tfinal$
is a $\mathtt{Missing}$ round and $\mathcal{O}^{\mathtt{Chasing}}$ posts the all-$1$ price vector in round $t$, then
$\phi(t + 1) < \phi(t)$;
and
(III)
if
$\phi(t) = 0$
for some
$\tinit \leq t \leq \tfinal$,
then
$\phi(t') = 0$
for all
$t < t' \leq \tfinal$.
We conclude that if $\mathcal{O}^{\mathtt{Chasing}}$ posts the all-$1$ price vector in 
$C W$
contiguous $\mathtt{Missing}$ rounds, then $\phi(\cdot)$ must reach zero and following that, there are no more $\mathtt{Missing}$ rounds.
Therefore the total number $\#\mathtt{M}$ of $\mathtt{Missing}$ rounds is stochastically dominated by
$C W$
times a geometric random variable $Z$ with parameter
$\epsilon^{C W}$.
Since
$\mathbb{E}[Z] = \epsilon^{-C W}$,
it follows that
$\mathbb{E}[\#\mathtt{M}] \leq C W / \epsilon^{C W}$.
Combined with the $\mathtt{Following}$ rounds, the CR is upper bounded by
$\epsilon \cdot T + C W / \epsilon^{C W}$.
The assertion is established by setting
$\epsilon = (T / (C W))^{-1 / (C W + 1)}$.
%
%Consider the same potential function $\phi(t)$. The evolution of $\phi(t)$ in all cases is the same, with a slight difference that in a round $t\in\mathtt{Missing}$ that algorithm's coin flips tails, $\phi(t)$ can actually increase (by at most $W$, due to extra units of resources in $\mathtt{Good}_t$ that our algorithm can potentially sell compared to target policy $\gamma$). To analyze this situation, we consider a slightly different stochastic process. This process starts at level $CW$. At each time, with probability $\epsilon$ decreases by $1$, and with the remaining probability restarts at level $CW$. Suppose stopping time $\tau$ is the first time this process hits $0$. Clearly, $\lvert\mathtt{Missing}\rvert\leq \tau$. To bound $\mathbb{E}[\tau]$, note that $\tau$ is a geometric random variable with success probability $\epsilon^{CW}$, and hence $\mathbb{E}[\tau]=(1/\epsilon)^{CW}$. Therefore, CR is bounded above by $\epsilon\cdot T+(1/\epsilon)^{CW}$. Setting $\epsilon=T^{-1/(CW+1)}$ finishes the proof. 
\end{proof}

\subsection{Putting the Pieces Together: Reduction to Online Learning with Switching Cost}
\label{section:off-the-shelf-SOLD}
Having an ongoing chasing oracle with vanishing chasing regret in hand, our remaining key technical idea is to reduce online decision making for the Dd-MDP problem to the well-studied problem of \emph{online learning with switching cost (OLSC)}~\cite{Kailai2005EAO}.
The problem's setup under full-information is exactly the same as the classic problem of learning from experts' advice, but the learner incurs an extra cost $\Delta > 0$,
a parameter referred to as the \emph{switching cost}, whenever it switches from one expert to another.
Here, we have a finite set $\Gamma$ of experts (often called actions or arms) and $T \in \mathbb{Z}_{> 0}$ rounds.
The expert reward function
$F_{t} : \Gamma \mapsto [0, 1)$
is revealed as feedback at the end of round
$t = 1, \dots, T$.
The goal of an algorithm for this problem is to pick a sequence
$\gamma_1, \ldots, \gamma_T$
of experts in an online fashion with the objective of minimizing the regret, now defined to be
\[\textstyle
\max_{\gamma \in \Gamma} \sum_{t \in [T]} F_{t}(\gamma)
-
\left( \sum_{t \in [T]} \mathbb{E} \left[ F_{t}(\gamma_{t}) \right]
-
\Delta \cdot \sum_{t = 2}^{T} 1_{\gamma_{t} \neq \gamma_{t - 1}} \right) \, .
\]

% there is a set of actions/arms/experts $\Gamma$, and  $T\in\mathbb{N}$ rounds. At each {\Step} $t \in [T]$, an adversary specifies a reward function $F_{t}: \Gamma \mapsto [0, 1]$, which is unknown to the online algorithm in the beginning of this round. Simultaneously, the algorithm chooses an action $\gamma_t\in\Gamma$.  Then the reward function $F_{t}(\cdot)$ is revealed to the algorithm. The goal of the algorithm is to pick a sequence of actions $\gamma_1,\ldots, \gamma_T$ in an online fashion to maximize $\mathbb{E}\Big[\sum_{t \in [T]}F_{t}(\gamma_{t}) \Big] - \Delta \cdot \sum_{t=2}^{T} 1_{\gamma_{t} \neq \gamma_{t - 1}} $, where $\Delta\cdot \sum_{t=2}^{T} 1_{\gamma_{t} \neq \gamma_{t - 1}}$ is essentially the extra cost that the algorithm incurs due to switching its actions and the parameter $\Delta>0$ is referred to as the \emph{switching cost}. The regret is defined to be
% \begin{equation*}
% \max_{\gamma \in \Gamma}\sum_{t \in [T]}F_{t}(\gamma) - \sum_{t \in [T]}\mathbb{E}\Big[F_{t}(\gamma_{t})\Big] +  \Delta \cdot \sum_{t = 2}^{T} 1_{\gamma_{t} \neq \gamma_{t - 1}} \, .
% \end{equation*}

\begin{theorem}[\cite{Kailai2005EAO}]\label{lemma_preliminary_on_ODSC}
The OLSC problem with switching cost $\Delta$ admits an online algorithm $\mathcal{A}$ whose regret is
$O \left( \sqrt{\Delta \cdot T \log |\Gamma|} \right)$.
\end{theorem}
\vspace{-2mm}
Note that the same theorem also holds for independent stochastic switching costs with $\Delta$ as the upper bound on the expected switching cost, simply because of linearity of expectation and the fact that in algorithms for OLSC, such as the Following-The-Perturbed-Leader \cite{Kailai2005EAO}, switching at each time is independent of the realized cost of switching. 

%%%%%%%%%%%%%%%%%%%%%%%%%%%%%%%%%%%%%%%

%%%%%%%%%%%%%%%%%%%%%%%%%%%%%%%%%%%%%%%
We now present our full-information online learning algorithm for $\sigma$-chasable Dd-MDP instances;
the reader is referred to \Cref{section:bandit} for the bandit version of this algorithm.
Our (full-information) algorithm, called \emph{chasing and switching (C\&S)}, requires a black box access to an algorithm $\mathcal{A}$ for the OLSC problem with the following configuration:
(1)
the expert set of $\mathcal{A}$ is identified with the policy collection $\Gamma$ of the Dd-MDP instance;
(2)
the number of rounds of $\mathcal{A}$ is equal to the number of rounds of the Dd-MDP instance ($T$);
and
(3)
the switching cost of $\mathcal{A}$ is set to
$\Delta = \sigma$.

The operation of C\&S is described in Algorithm~\ref{alg:chasing_and_switching}. This algorithm maintains, in parallel, the OLSC algorithm $\mathcal{A}$ and an ongoing chasing oracle $\mathcal{O}^{\mathtt{Chasing}}$; $\mathcal{A}$ produces a sequence
$\{ \gamma_{t} \}_{t = 1}^{T}$
of policies and $\mathcal{O}^{\mathtt{Chasing}}$ produces a sequence
$\{ x_{t} \}_{t = 1}^{T}$
of actions based on that.
Specifically, $\mathcal{O}^{\mathtt{Chasing}}$ is \emph{restarted}, i.e., invoked from scratch with a fresh policy
$\gamma$, whenever $\mathcal{A}$ switches to $\gamma$ from some policy
$\gamma' \neq \gamma$.

%It starts to run from the given state $s_{0}$ as follows. At the beginning of each {\Step} $t \in [T]$, C\&S invokes Algorithm $\mathcal{A}$ to choose a policy $\gamma_{t}$ from $\Gamma$. If $\gamma_{t}$ is different from the one chosen in the last step or $t = 1$, then C\&S starts a new run of $\mathcal{O}^{\mathtt{Chasing}}$ with $\gamma_{t}$ as the chased policy, $s_{t - 1}$ as the initial state and $t$ as the initial time, and takes the action generated by $\mathcal{O}^{\mathtt{Chasing}}$. Otherwise, C\&S simply follows the oracle $\mathcal{O}^{\mathtt{Chasing}}$. After C\&S makes its decision for {\Step} $t$, the reward function $f_{t}(\cdot, \cdot)$ is revealed. C\&S then computes $f_{t }(\FPS{t - 1}{\gamma}, \FPA{t}{\gamma})$ for every policy $\gamma \in \Gamma$ by simulation, and feeds $F_{t}(\gamma) = f_{t}(\FPS{t - 1}{\gamma}, \FPA{t}{\gamma})$ to Algorithm $\mathcal{A}$ as the reward of $\gamma$ at $t$.

%Note that the action $\FPA{t}{\gamma}$ and the reward $f_{t - 1}(\FPS{t - 2}{\gamma}, \FPA{t - 1}{\gamma})$ corresponding to every policy $\gamma \in \Gamma$ at step $t$ because we take the full information setting for the reward functions and the state transition functions. It is consistent to feed $f_{t - 1}(\FPS{t - 1}{\gamma}, \FPA{t}{\gamma})$ to $\mathcal{A}$ as the reward $F_{t - 1}(\gamma)$ of $\gamma$ at $t - 1$ because $f_{t - 1}(\FPS{t - 1}{\gamma}, \FPA{t}{\gamma})$ only relies on $t$ and $\gamma$.

\begin{algorithm}[t]
\caption{Online Dd-MDP algorithm C\&S}
\label{alg:chasing_and_switching}
	 \SetAlgoLined
	\KwIn{Policy set $\Gamma$, OLSC algorithm $\mathcal{A}$, chasing oracle $\mathcal{O}^{\mathtt{Chasing}}$, initial state $s_1$;}
	\KwOut{Sequence $x_1, \ldots, x_T$ of actions, (implicit) sequence $s_2, \ldots, s_T$ of states;}
	\vspace{1mm}
	Start from initial state $s_1$;\\
	\For{each {\Step} $t \in [T]$}{
	    \vspace{0.5mm}
		Invoke $\mathcal{A}$ to pick a policy $\gamma_{t}$ at the beginning of round $t$;\\
		\vspace{1mm}
		\eIf{$t > 1$ and $\gamma_{t} \neq \gamma_{t - 1}$}{
			%$flag = false$;\\ 
			%Stop current run of oracle $\mathcal{O}^{\mathtt{Chasing}}$, if still running;\\
			Invoke $\mathcal{O}^{\mathtt{Chasing}}$ from scratch with target policy $\gamma_{t}$, initialized with round $t$ and state $s_{t}$; \\
			Select the action $x_t \leftarrow \hat{x}(t)$ returned by $\mathcal{O}^{\mathtt{Chasing}}$;
		}{
			%\If{$\mathcal{O}^{\mathtt{Chasing}}$ is running}{
				Continue the existing run of $\mathcal{O}^{\mathtt{Chasing}}$ and select the action $x_t \leftarrow \hat{x}(t)$ it returns;
			%}
			%\Else{
				%Compute action $\FPA{t}{\gamma_{t}}$ by simulating policy $\gamma_{t}$ up to time $t$ and pick it;
			%}
		}
		\vspace{1mm}
		Feed $\mathcal{O}^{\mathtt{Chasing}}$ with $g_{t}(\cdot, \cdot)$ and $f_{t}(\cdot, \cdot)$ as the state transition and reward functions of round $t$; \\
	    \vspace{1mm}
		\For{each $\gamma \in \Gamma$}{
			Compute $F_{t}(\gamma) \gets f_{t }(\FPS{t}{\gamma}, \FPA{t}{\gamma})$ by simulating policy $\gamma$ up to time $t$ (see Eq.~\eqref{formula_fixed_policy_states_and_actions}); \\
		}
		\vspace{1mm}
		Feed $\mathcal{A}$ with $F_{t}(\cdot)$ as the reward function of round $t$;
	}
\end{algorithm}

\begin{theorem}\label{theorem_regret_of_hola}
The regret of C\&S for $T$-round $\sigma$-chasable Dd-MDP instances is
$O \left( \sqrt{\sigma \cdot T \log |\Gamma|} \right)$.
\end{theorem}
\vspace{-4mm}
\begin{proof}
Partition the $T$ {\Steps} into episodes
$\lbrace 1, 2, \dots \rbrace$
so that each episode $\theta$ is a maximal contiguous sequence of rounds in which the policy $\gamma_{\theta}$ chosen by $\mathcal{A}$ does not change.
Let $t_{\theta}$ and $t'_{\theta}$ be the first and last rounds of episode $\theta$, respectively.
Consider some episode $\theta$ with corresponding policy $\gamma_{\theta}$.
Since C\&S follows an action sequence generated by $\mathcal{O}^{\mathtt{Chasing}}$ during the round interval
$[t_{\theta}, t'_{\theta}]$ and since the chasing regret of $\mathcal{O}^{\mathtt{Chasing}}$ is upper bounded by
$\sigma = \Delta$,
it follows that
\[\textstyle
\sum_{t = t_{\theta}}^{t'_{\theta}} F_{t}(\gamma_{\theta})
-
\sum_{t = t_{\theta}}^{t'_{\theta}} \mathbb{E} \left[ f_{t}(s_{t}, x_{t}) \right]
\, = \,
\sum_{t = t_{\theta}}^{t'_{\theta}}
f_{t} \left( \FPS{t}{\gamma_{\theta}}, \FPA{t}{\gamma_{\theta}} \right)
-
\sum_{t = t_{\theta}}^{t'_{\theta}} \mathbb{E} \left[ f_{t}(s_{t}, x_{t}) \right]
\, \leq \,
\Delta \, .
\]
Therefore, for each policy
$\gamma \in \Gamma$, we have
\vspace{-3mm}
%\OriginalDisplay
%\begin{align*}
%\sum_{t \in [T]}f_{t}\Big( \FPS{t - 1}{\gamma}, \FPA{t}{\gamma} \Big) - \sum_{t \in [T]}\mathbb{E}\Big[ f_{t}(s_{t - 1}, x_{t})\Big] \leq& \sum_{t \in [T]}f_{t}\Big( \FPS{t - 1}{\gamma}, \FPA{t}{\gamma} \Big) - \sum_{\theta}\bigg[\sum_{t = t_{\theta}}^{t_{\theta}'}\mathbb{E}\Big[ F_{t}(\gamma_{t_{\theta}}) \Big] - \Delta \bigg] \\
%=& \sum_{t \in [T]}F_{t}(\gamma) - \bigg[\sum_{t \in [T]}\mathbb{E}\Big[ F_{t}(\gamma_{t}')\Big] - \Delta \cdot \sum_{t = 2}^{T}1_{\gamma_{t}' \neq \gamma_{t - 1}'} \bigg] \, ,
%\end{align*}
%\OriginalDisplayEnd
%\MultiLineSplit
\begin{align*}
\sum_{t \in [T]} f_{t} \left( \FPS{t}{\gamma}, \FPA{t}{\gamma} \right)
-
\sum_{t\in [T]}\mathbb{E} \left[ f_{t}(s_{t}, x_{t}) \right] 
&\leq
\sum_{t \in [T]} f_{t} \left( \FPS{t}{\gamma}, \FPA{t}{\gamma} \right)
-
\sum_{\theta} \left(
\sum_{t = t_{\theta}}^{t'_{\theta}}
\mathbb{E} \left[ F_{t}(\gamma_{\theta}) \right]
- \Delta \right) \\
&= 
\sum_{t \in [T]} F_{t}(\gamma)
-
\left(
\sum_{t \in [T]} \mathbb{E} \left[ F_{t}(\gamma_{t}) \right]
-
\Delta \cdot \sum_{t = 2}^{T} 1_{\gamma_{t} \neq \gamma_{t - 1}} \right) .
\end{align*}
By \Cref{lemma_preliminary_on_ODSC}, the last expression is at most
$O \left( \sqrt{\Delta \cdot T \log |\Gamma|} \right)
=
O \left( \sqrt{\sigma \cdot T \log |\Gamma|} \right)$.
% Note that although the adversary in DRACC is allowed to be \emph{adaptive}, which means that the adversary is aware of $\lbrace x_{t'} \rbrace_{t' \in [t - 1]}$ when deciding $f_{t}(\cdot)$ and $g_{t}(\cdot)$, the function $F_{t}(\gamma) = f_{t}(\FPS{t - 1}{\gamma}, \FPA{t}{\gamma})$ does not depend on $\lbrace \gamma_{t'} \rbrace_{t' \in [t - 1]}$. This ensures that it is feasible to apply the regret guaranteed by \Cref{lemma_preliminary_on_ODSC}. 
\end{proof}
\vspace{-3mm}
So far, we have only considered the notion of \emph{policy regret} as defined in \cref{formula_policy_regret}.
An extension of our results to the notion of \emph{external regret} \citep{Arora2018PRR} is discussed in \Cref{section:external_regret}.
Furthermore, we investigate the bandit version of the problem in \Cref{section:bandit}.
In a nutshell, by introducing a stateless version of our full-information chasing oracle and reducing to the adversarial multi-armed-bandit problem~\citep{Audibert2009MPA}, we obtain $O(T^{2/3})$ regret bound for Dd-MDP under bandit feedback.
Finally, we obtain near-matching lower bounds for both the full-information and bandit feedback versions of the Dd-MDP problem under the chasability condition in \Cref{apx:lower}.

%%%%%%%%%%%%%%%%%%%%%%%%%%%%%%%%%%%%%%%%%%%%

\subsubsection*{\emph{Relation to the DRACC Problem (continued)}}

%%%%%%%%%%%%%%%%%%%%%%%%%%%%%%%%%%%%%%%%%%%%%%%%%%%%%%%%%%%%%%%%%%%%%%%%%%%%%%
We can now use C\&S (Algorithm~\ref{alg:chasing_and_switching}) for Dd-MDPs that correspond to DRACC instances.
This final mechanism is called \emph{learning based posted pricing ({\MechanismName})}.
It first provides the input parameters of C\&S, including the collection $\Gamma$ of pricing policies, the OLSC algorithm $\mathcal{A}$, the ongoing chasing oracle $\mathcal{O}^{\mathtt{Chasing}}$ and the initial state $s_{1}$.
It then runs C\&S by posting its price vectors (actions) and updating the resulting inventory vectors (states).
For $\mathcal{O}^{\mathtt{Chasing}}$, we employ the (randomized) chasing oracles promised in \Cref{theorem_chasing_condition}  and \Cref{theorem_chasing_condition-general}.
The following theorems can now be inferred from \Cref{theorem_regret_of_hola},  \Cref{theorem_chasing_condition}, and \Cref{theorem_chasing_condition-general}.

\begin{theorem}\label{theorem_regret_of_mechanism_for_BIT}
The regret of {\MechanismName} for $T$-round DRACC instances with $k_{t}$-demand valuation functions (or more generally, with the valuation functions defined in \Cref{remark:generalize-k-demand}) is
$O \left( (C W)^{\frac{1}{4}} T^{\frac{3}{4}} \sqrt{\log |\PricePols|} \right)$.
\end{theorem}

\begin{theorem}\label{theorem_regret_of_mechanism_for_BIT_general}
The regret of {\MechanismName} for $T$-round DRACC instances with with arbitrary valuation functions is
$O \left(
T^{\frac{1}{2} \left( 1 + \frac{C W}{C W + 1} \right)} \sqrt{\log |\PricePols|} \right)$.
\end{theorem}

Note that the regret bounds in \Cref{theorem_regret_of_mechanism_for_BIT} and \Cref{theorem_regret_of_mechanism_for_BIT_general} depend on the parameters $C$ and $W$ of the DRACC problem;
as shown in the following theorem, such a dependence is unavoidable.

%Notably, the guarantee that {\MechanismName} has vanishing regret requires the modeling assumption that $\MaxMultiSym \cdot W \cdot \log^{2} |\Gamma| = o(T)$.
%By contrast, no posted pricing mechanism can guarantee a sublinear regret when $\MaxMultiSym \cdot W = \Omega(T)$, as we prove in \Cref{section:assumption-necessity} in the appendix. See \Cref{theorem_impossibility_result_for_C_W} for more details.

%\label{section:assumption-necessity}

%In this section, we justify the necessity of our modeling assumptions by proving that the vanishing regret is impossible if $C \cdot W$ grows linearly with $T$.

%In this section, we prove that the vanishing regret is impossible if $C \cdot W$ grows linearly with $T$.

\begin{theorem}\label{theorem_impossibility_result_for_C_W}
If $C \cdot W = \Omega(T)$,
then the regret of any posted price mechanism is
$\Omega(T)$.
\end{theorem}
\begin{proof}
Here we construct two instances of DRACC. The following settings are the same between these two instances.
\begin{itemize}
\item The parameters $\MaxMultiSym$ and $W$ are chosen so that $C \cdot W = \frac{T}{2}$. Set $N = W$.
\item For each {\Item} $i$, $t_{a}(i) = 1$ and $t_{e}(i) = T$. This setting implies that for every {\Buyer} $t$, $A_{t} = [N]$, which is consistent with $W = N$. Every $i \in [N]$ has the same {\Multiplicity} $\Capacity(i) = \MaxMultiSym$.
\item For each {\Buyer} $t \in \Big[1, \frac{T}{2} \Big]$, the valuation function $v_{t}$ is set as follows.
\begin{equation*}
v_{t}(A') = \begin{cases}
\frac{1}{2} & \text{if } |A'| = 1 \\
0 & \text{otherwise}
\end{cases} \;\;\;\; \forall A' \subseteq A_{t} \, .
\end{equation*}
\end{itemize}

For the {\Buyers} $t \in \Big[\frac{T}{2} + 1, T \Big]$, their valuation functions are different between the two instances. In particular, in the first instance, $v_{t}(A') = 0$ for any $A' \subseteq A_{t}$, while in the second instance
\begin{equation*}
v_{t}(A') = \begin{cases}
1 - \epsilon & \text{if } |A'| = 1 \\
0 & \text{otherwise}
\end{cases} \;\;\;\; \forall A' \subseteq A_{t} \, .
\end{equation*}
where $\epsilon$ is some small enough constant in $(0, \frac{1}{2})$.

Now consider an arbitrary deterministic mechanism $\mathcal{M}$. Such a mechanism will output the same sequence of price vectors for the first half of the {\Buyers} in these two instances. Therefore, the total number of {\Items} that are allocated by $\mathcal{M}$ to the first half of {\Buyers} must be the same $k$ in the two instances for some $k \in \Big[0, \frac{T}{2} \Big]$. Then, the revenue of $\mathcal{M}$ is at most $\frac{k}{2}$ in the former instance, while at most $\frac{k}{2} + \Big(\frac{T}{2} - k \Big) \cdot (1 - \epsilon) = \frac{1 - \epsilon}{2}T - (\frac{1}{2} - \epsilon)k$ in the latter one. Now consider a pricing policy $\gamma$ that maps every {\Surplus} vector except $\LVEC 0 \RVEC$ to a price vector that only contains $\frac{1}{2}$. The revenue of $\gamma$ in the first instance is $\frac{T}{4}$. Similarly, there exists a policy $\gamma'$ with revenue $\frac{T}{2} \cdot (1 - \epsilon)$ in the second instance. Therefore, the regret of $\mathcal{M}$ is at least
\begin{equation*}
\max\bigg\lbrace \frac{T}{4} - \frac{k}{2},\;\; \frac{T}{2} (1 - \epsilon) - \Big[ \frac{1 - \epsilon}{2}T - (\frac{1}{2} - \epsilon)k \Big] \bigg\rbrace \geq \frac{1 - 2\epsilon}{8(1 - \epsilon)} T \, .
\end{equation*}

To generalize the result above to the mechanisms that can utilize the random bits, here we adopt Yao's principle \cite{Yao1997PCM}. In particular, we construct a distribution over the inputs which assigns probabilities $\frac{1 - 2\epsilon}{2 - 2\epsilon}$ and $\frac{1}{2 - 2\epsilon}$ to the two instances constructed above, respectively. It can be verified that against such a distribution, the expectation of any random mechanism's regret is at least $\frac{1 - 2\epsilon}{8(1 - \epsilon)}T$. By Yao's principle, the lower bound on the regret of any mechanism that can utilizes the random bits is also $\frac{1 - 2\epsilon}{8(1 - \epsilon)}T$. Therefore, this proposition is established.
\end{proof}

\section{Applications of the DRACC Problem}
\label{sec:applications}
\label{section:applications}
The mechanism {\MechanismName} proposed for the DRACC problem can be directly applied to a large family of online pricing problems arising in practice.
Two examples are presented in this section: the \emph{online job scheduling (OJS)} problem and the problem of \emph{matching over dynamic bipartite graphs (MDBG)}.

%%%%%%%%%%%%%%%%%%%%%%%%%%%%%%%%%%%%%%%
\subsection{Online Job Scheduling}\label{subsection_online_job_scheduling}
%%%%%%%%%%%%%%%%%%%%%%%%%%%%%%%%%%%%%%%
The OJS problem described in this section is motivated by the application of assigning jobs that arrive online to limited bandwidth {\Slots} for maximizing the total payments collected from the jobs.
Formally, in the OJS problem, there are $T$ strategic myopic jobs, arriving sequentially over $N$ time slots.
Each {\Slot}
$i \in [N]$
lasts over the time interval
$[i, i + 1)$
and is associated with a bandwidth $c(i)$, which means that this slot can be allocated to at most $c(i)$ jobs.
For each job
$t \in T$,
the adversary specifies an arrival slot
$1 \leq a_{t} \leq N$,
a departure slot
$a_{t} \leq d_{t} \leq N$,
a length
$1 \leq l_{t} \leq d_{t} - a_{t} + 1$,
and a value
$v_{t} \in [0, 1)$.
We emphasize that any number (including zero) of jobs may have slot $i$ as their arrival (or departure) slot.
The goal of job $t$ is to get an allocation of $l_{t}$ contiguous slots within
$[a_{t}, d_{t}]$,
namely, a slot interval in
\[
\mathcal{I}_{t}
\, = \,
\left\{ [i, i + l_{t} - 1] \mid a_{t} \leq i \leq d_{t} - l_{t} + 1 \right\} \, ,
\]
with $v_{t}$ being the job's value for each such allocation.
Let $C$ and $W$ be upper bounds on
$\max_{i \in [N]} c(i)$
and
$\max_{t \in [T]} d_{t} - a_{t} + 1$,
respectively.

%The interval $[a_{t}, d_{t}]$ is denoted by $A_{t}$. The maximum number of slots in $A_{t}$ over every $t \in T$ is denoted by $W$.

Job
$t \in [T]$
is reported to the OJS mechanism at the beginning of slot $a_{t}$;
if several jobs share the same arrival slot, then they are reported to the mechanism sequentially in an arbitrary order.
At the beginning of slot $a_{t}$, the mechanism is also informed of the bandwidth parameter $c(i)$ of every slot
$i \in A_{t}$,
where
$A_{t}$
is defined to be the slot interval
\[
\Active_{t}
\, = \,
[a_{t}, a_{t} + W - 1] \, ;
\]
note that the mechanism may have been informed of the bandwidth parameters of some slots in $\Active_{t}$ beforehand (if they belong to $\Active_{t'}$ for
$t' < t$).
In response, the mechanism posts a price vector
$\PriceVec_{t} \in (0, 1]^{\Active_{t}}$
and elicits the parameters
$d_{t}$, $l_{t}$, and $v_{t}$.
Subsequently, (one bandwidth unit of) the slots in the demand set $\hat{A}_{t}^{\PriceVec_{t}}$ are allocated to job $t$ at a total price of $\hat{\Pay}_{t}^{\PriceVec_{t}}$, where
\[
\hat{A}_{t}^{\PriceVec}
\, = \,
\begin{cases}
\emptyset & \text{if } v_{t} < \min_{I \in \mathcal{I}_{t}} \sum_{i \in I} \PriceVec(i) \\
\argmin_{I \in \mathcal{I}_{t}} \sum_{i \in I} \PriceVec(i) & \text{otherwise}
\end{cases}
\qquad \text{and} \qquad
\hat{\Pay}_{t}^{\PriceVec}
\, = \,
\sum_{i \in \hat{A}_{t}^{\PriceVec}} \PriceVec(i)
\]
for any price vector
$\PriceVec \in (0, 1]^{\Active_{t}}$,
consistently breaking $\argmin$ ties according to the lexicographic order on $\Active_{t}$.

Let
$\SurplusVec_{t} \in \{ 0, 1, \dots, C \}^{\Active_t}$
be the (remaining) bandwidth vector that encodes the number $\SurplusVec_{t}(i)$ of units remaining from the bandwidth of slot
$i \in \Active_{t}$
before processing job
$t = 1, \dots, T$.
Formally, if slot $i$ has not been allocated to any of the jobs in
$\{ 1, \dots, t - 1 \}$,
then
$\SurplusVec_{t}(i) = \Capacity(i)$;
and if (a bandwidth unit of) slot $i$ is allocated to job $t$ and
$i \in A_{t + 1}$,
then
$\SurplusVec_{t + 1}(i) = \SurplusVec_{t}(i) - 1$.
We say that a price vector $\PriceVec$ is feasible for the bandwidth vector $\SurplusVec_{t}$ if
$\PriceVec(i) = 1$
for every
$i \in \Active_{t}$
such that
$\SurplusVec_{t}(i) = 0$,
that is, for every slot $i$ that has already been exhausted before job $t$ is processed.
To ensure that the slots' bandwidth is not exceeded, we require that the posted price vector $\PriceVec_{t}$ is feasible for $\SurplusVec_{t}$ for every
$1 \leq t \leq T$.
We aim for posted price OJS mechanisms whose objective is to maximize the total expected payment
$\Ex[\sum_{t = 1}^{T} \hat{\Pay}_{t}^{\PriceVec_{t}}]$
received from all jobs, where the expectation is over the mechanism's internal randomness.

A pricing policy $\PricePol$ is a function that maps each bandwidth vector
$\SurplusVec \in \{ 0, 1, \dots, C\}^{A_{t}}$,
$t \in [T]$,
to a price vector
$\PriceVec = \PricePol(\SurplusVec)$,
subject to the constraint that $\PriceVec$ is feasible for $\SurplusVec$.
Given a pricing policy $\PricePol$, consider a decision maker that repeatedly plays according to $\PricePol$;
namely, she posts the price vector
$\PriceVec^{\PricePol}_{t} = \PricePol(\SurplusVec^{\PricePol}_{t})$
for job
$t = 1, \dots, T$,
where $\SurplusVec^{\PricePol}_{t}$ is the bandwidth vector obtained by applying $\PricePol$ recursively on previous bandwidth vectors $\SurplusVec^{\PricePol}_{t'}$ and posting prices
$\PricePol(\SurplusVec^{\PricePol}_{t'})$
for jobs
$t' = 1, \dots, t - 1$.
Denoting
$\hat{\Pay}_{t}^{\PricePol} = \hat{\Pay}_{t}^{\PriceVec^{\PricePol}_{t}}$,
the revenue of this decision maker is given by
$\sum_{t = 1}^{T} \hat{\Pay}_{t}^{\PricePol}$.
Given a collection $\Gamma$ of pricing policies, the quality of a posted price OJS mechanism
$\{ \PriceVec_{t} \}_{t = 1}^{T}$
is measured by means of the decision maker's regret with respect to $\Gamma$, namely
\[
\max_{\PricePol \in \PricePols}
\sum_{t = 1}^{T} \hat{\Pay}_{t}^{\PricePol}
-
\Ex \left[
\sum_{t = 1}^{T} \hat{\Pay}_{t}^{\PriceVec_{t}}
\right]
\, ,
\]
where the expectation is taken over the decision maker's randomness.

\PARA{Reduction to DRACC.}
Given the aforementioned choice of notation, the transformation of an OJS instance to a DRACC instance should now be straightforward.
Specifically:
job $t$ is mapped to user $t$;
slot $i$ is mapped to resource $i$;
slot $i$'s bandwidth parameter $c(i)$ is mapped to the capacity of resource $i$;
job $t$'s arrival slot $a_{t}$ determines the set $A_{t}$ of active resources at time $t$, and through these sets, the arrival and departure times of the resources;
and
job $t$'s length $l_{t}$ and value $v_{t}$ parameters determine the valuation function of user $t$, assigning a value of $v_{t}$ to each
$I \in \mathcal{I}_{t}$;
and a zero value to any other subset of $A_{t}$.
The following corollary is now inferred directly from \Cref{theorem_regret_of_mechanism_for_BIT_general}.

\begin{corollary} \label{corollary_regret_of_JS}
The OJS problem admits a mechanism whose regret for $T$-round instances is
$O \left( T^{\frac{1}{2}\left(1+\frac{\MaxMultiSym W}{\MaxMultiSym W + 1}\right)} \sqrt{\log |\PricePols|} \right)$.
\end{corollary}

\Cref{corollary_regret_of_JS} is derived from the regret bound of {\MechanismName} for the DRACC problem with arbitrary valuation functions, based on the (randomized) chasing oracle implementation developed in \Cref{theorem_chasing_condition-general}.
It turns out though that one can exploit the structural properties of the OJS problem to design a chasing oracle with dramatically improved chasing regret, thus improving the regret bound for the OJS problem (see \cref{corollary_improved_regret_OJS}).

\begin{lemma} \label{lemma_chasing_job_scheduling}
The OJS problem admits a (deterministic) ongoing chasing oracle whose chasing regret is at most
$2 \cdot \MaxMultiSym W$.
\end{lemma}

\begin{proof}
One property of OJS is that for any two {\Slots} $i$ and $i'$,
\begin{equation}
t_{a}(i) \leq t_{a}(i')
\quad\Rightarrow \quad
t_{e}(i) \leq t_{e}(i') \, .
\label{formula_OJS_FIFO}
\end{equation}
We prove the claim by constructing a chasing ongoing oracle $\mathcal{O}^{\mathtt{Chasing}}$ with the desired CR using this property. Given a target policy $\gamma$, an initial step $\tinit$, and an initial state $\sinit$, oracle $\mathcal{O}^{\mathtt{Chasing}}$ posts a price vector $\hat{\PriceVec}_{t}$ for each $t \geq \tinit$ as follows
\begin{equation*}
\hat{\PriceVec}_{t} = \begin{cases}
	\langle 1 \rangle_{i \in \Active_{t}} & \text{if } t \leq \min\lbrace T, \max_{i \in \Active_{\tinit}} t_{e}(i) \rbrace \\
	\FPprice{t}{\gamma} & \text{otherwise}
\end{cases} \, .
\end{equation*}
Let $t' = \min\lbrace T, \max_{i \in \Active_{\tinit}} t_{e}(i) \rbrace$. The price vector $\hat{\PriceVec}_{t}$ is trivially feasible for every $t \in [\tinit, t']$. If $t' < T$, then for every {\Slot} $i$ in $A_{t' + 1}$, we have $i \notin \Active_{\tinit}$, which gives $t_{a}(i) > \tinit$. Since $\mathcal{O}^{\mathtt{Chasing}}$ does not sell any {\Slot} to {\Buyers} from $\tinit$ to $t'$, it holds that $\hat{\VEC{\SurplusSym}}_{t' + 1}(i) = \MultiSym{i} \geq \FPremain{t' + 1}{\gamma}(i)$. Therefore, $\mathtt{Good}_{t' + 1} = \Active_{t' + 1}$ and $\mathtt{Bad}_{t' + 1} = \emptyset$. 
Then it can be proved inductively that for any $t \geq t' + 1$, $\mathtt{Bad}_{t} = \emptyset$, which ensures the feasibility of $\hat{\PriceVec}_{t}$. 
Moreover, for each $t \geq t' + 1$, since $\hat{\PriceVec}_{t} = \FPprice{t}{\gamma}$, we have $f_{t}\big( \hat{\VEC{\SurplusSym}}_{t}, \hat{\PriceVec}_{t} \big) = f_{t}\big( \FPremain{t}{\gamma},  \FPprice{t}{\gamma} \big)$.

It remains to bound $\sum_{t = \tinit}^{t'} f_{t}\big( \FPremain{t}{\gamma},  \FPprice{t}{\gamma} \big) - \sum_{t = \tinit}^{t'} f_{t}\big( \hat{\VEC{\SurplusSym}}_{t}, \hat{\PriceVec}_{t} \big)$. Let $S$ be the set of {\Slots} $i$ with $t_{a}(i) \in (\tinit, t']$. By Eq.~\eqref{formula_OJS_FIFO}, it holds for every $i \in S$ that $t_{e}(i) \geq t'$, because for every $i' \in \Active_{\tinit}$, $t_{a}(i') \leq \tinit$. By definition, $S \subseteq \Active_{t'}$, which gives $|S| \leq W$. Since the {\Slots} that can be sold by any policy to {\Buyers} in $[\tinit, t']$ belong to $A_{\tinit} \cup S$, we have 
\begin{equation*}
\sum_{t = \tinit}^{t'} f_{t}\big( \FPremain{t}{\gamma},  \FPprice{t}{\gamma} \big) \leq \MaxMultiSym \cdot |A_{\tinit} \cup S| \leq \MaxMultiSym \cdot 2 W \, .
\end{equation*}
Since $\sum_{t = \tinit}^{t'} f_{t}\big( \hat{\VEC{\SurplusSym}}_{t}, \hat{\PriceVec}_{t} \big)$ is non-negative, this theorem is established.
\end{proof}

By plugging \Cref{lemma_chasing_job_scheduling} into \Cref{theorem_regret_of_hola}, we obtain the following improvement to \cref{corollary_regret_of_JS};
this bound is near-optimal due to \cite{blum2005near}.

\begin{corollary} \label{corollary_improved_regret_OJS}
The OJS problem admits a mechanism whose regret for $T$-round instances is
$O \left( \sqrt{\MaxMultiSym W \cdot T \log |\PricePols|} \right)$.
\end{corollary}

%%%%%%%%%%%%%%%%%%%%%%%%%%%%%%%%%%%%%%%
\subsection{Matching Over Dynamic Bipartite Graphs}
%%%%%%%%%%%%%%%%%%%%%%%%%%%%%%%%%%%%%%%
The MDBG problem is a dynamic variation of the conventional bipartite matching problem with the goal of maximizing the revenue. Formally, 
in the MDBG problem, there are two sets of nodes, the \emph{left-side} node set $\LeftSet = \lbrace \LeftNode \rbrace_{\LeftNode \in \LeftInts}$ and the \emph{right-side} node set $\RightSet = \lbrace \RightNode \rbrace_{\RightNode \in \RightInts}$. 
The nodes in each of these two sets arrive sequentially and dynamically. 
For each node $\LeftNode \in \LeftSet$, an adversary specifies a pair of parameters $t_{a}(\LeftNode) \in \RightInts$ and $t_{e}(\LeftNode) \in \Big[t_{a}(\LeftNode), \RightInts \Big]$. It means that the node $\LeftNode$ arrives just before the arrival of the node $\RightNode = t_{a}(\LeftNode) \in \RightSet$, and expires immediately after the node ${\RightNode}' = t_{e}(\LeftNode) \in \RightSet$ is given. 
For each node $\RightNode \in \RightSet$, define $A_{\RightNode} = \lbrace \LeftNode \in \LeftSet: \RightNode \in [t_{a}(\LeftNode), t_{e}(\LeftNode)] \rbrace$. The adversary also specifies a weight $w_{\RightNode}(\LeftNode) \in [0, 1)$ for each $\RightNode \in \RightSet$ and $\LeftNode \in A_{\RightNode}$.

A posted price mechanism is required to present a price vector $\VEC{p}_{\RightNode} \in (0, 1]^{|A_{\RightNode}|}$ independently of $w_{\RightNode}(\cdot)$ upon the arrival of each node $\RightNode \in \RightSet$. For any price vector $\VEC{p}$ presented to $\RightNode$, define $\hat{A}_{\RightNode}^{\VEC{p}} = \ArgMax_{\LeftNode \in A_{\RightNode}}w_{\RightNode}(\LeftNode) - \VEC{p}(\LeftNode)$ with breaking ties in a fixed way. The mechanism matches the left-side node $\hat{A}_{\RightNode}^{\VEC{p}}$ to the right-side node $\RightNode$ and charges $\RightNode$ the payment $\VEC{p}\Big( \hat{A}_{\RightNode}^{\VEC{p}}\Big)$ if $w_{\RightNode}\Big( \hat{A}_{\RightNode}^{\VEC{p}} \Big) \geq \VEC{p}\Big( \hat{A}_{\RightNode}^{\VEC{p}} \Big)$. Otherwise, no left node is matched to $\RightNode$, and no payment is obtained. After that, $w_{\RightNode}(\cdot)$ is revealed to the mechanism.

In the MDBG problem, every left-side node $\LeftNode$ can only be matched to at most one right-side node $\RightNode$. We express this constraint as a feasibility requirement on the price vector that for each right-side node $\RightNode$, if a left-side node $\LeftNode \in A_{\RightNode}$ has already been matched before the arrival of $\RightNode$, then the price of $\LeftNode$ should be set to $1$. The states of whether the left-side nodes in $A_{\RightNode}$ have been matched can be described with a Boolean vector of length $|A_{\RightNode}|$, and a pricing policy $\gamma \in \Gamma$ is a mapping from each possible Boolean vector to a feasible price vector. The objective of the MDBG problem is to find a feasible price vector $\VEC{p}_{\RightNode}$ for every $\RightNode \in \RightSet$ to maximize the total payments, and the regret is defined to be the difference between the revenue obtained by the best fixed in-hindsight pricing policy in a given collection $\Gamma$ and the expected revenue of the mechanism. %For the MDBG problem, we make the assumption that the prices are taken from a finite set $\mathcal{P} \subset (0, 1]$ with $1 \in \mathcal{P}$.

\paragraph{Reduction to DRACC.} This problem can be transformed to a special case of the DRACC problem by taking the nodes in $\LeftSet$ (resp.~$\RightSet$) as the {\Items} (resp.~{\Buyers}). The {\Multiplicity} of every {\Item} is exactly one.  The valuation function of each {\Buyer} $\RightNode$ maps each subset $A \subseteq A_{\RightNode}$ to $v_{\RightNode}(A) = \max_{\LeftNode \in A} w_{\RightNode}(\LeftNode)$. Such a setting is consistent because the price posted for each {\Item} $\LeftNode$ is strictly larger than $0$, which ensures that at most one {\Item} is allocated to each {\Buyer}. Moreover, $v_{\RightNode}$ is a $k_{t}$-demand valuation function with $k_{\RightNode} = 1$ for every $\RightNode$. Using Theorem \ref{theorem_regret_of_mechanism_for_BIT}, we get the following result.

\begin{corollary}
For the MDBG problem, the regret of the mechanism {\MechanismName} is bounded by $O\Big( W^{\frac{1}{4}} T^{\frac{3}{4}} \sqrt{\log |\PricePols|}\Big)$.
\end{corollary}

%\clearpage

%\section*{Broader Impact}
%\input{tex/neur-impact.tex}

\clearpage
\bibliographystyle{plainnat}
\bibliography{references}

\clearpage
\appendix

\begin{figure}
    \centering
    \Large{\textbf{APPENDIX}}
\end{figure}

%\section{Further Related Work}
%\label{apx:further}
%\input{tex/further-related-work.tex}

%\section{Missing Proofs from \Cref{sec:main-technical}}
%\label{apx:missing-proofs}
%\input{tex/missing-proofs}

%\section{Applications of the DRACC Problem}
%\label{sec:applications}
%\input{tex/applications.tex}

\section{External Regret}
\label{section:external_regret}
%%%%%%%%%%%%%%%%%%%%%%%%%%%%%%%%%%%%%%%%%%%%%%%%%%%%%%%%%%%%%%%%%%%%%%%%%%%%%%

To complement our result, in this part we consider another natural alternative definition for regret, known as \emph{external regret}, defined as follows (see \cite{Arora2018PRR} for more details).

\begin{equation}
\max_{\gamma \in \Gamma} \sum_{t \in [T]} f_{t}\Big(s_{t}, \gamma(s_{t}) \Big) - \sum_{t \in [T]} \mathbb{E}\Big[ f_{t}(s_{t}, x_{t}) \Big] \, . \label{formula_external_regret}
\end{equation}
In words, while policy regret is the difference between the simulated reward of the optimal fixed policy and the actual reward of the algorithm, in external regret the reward that is being accredited to the optimal fixed policy in each round $t$ is the reward that policy would have obtained when being in the actual state of the algorithm (versus being in its simulated current state). %In this section,  we show for Dd-MDP with the Chasability condition, obtaining both sublinear external regret and sublinear policy regret is impossible. 
In \cite{Arora2018PRR}, it is shown that for the online learning problems where the reward functions depend on the $m$-recent actions, the policy regret and the external regret are incomparable, which means that any algorithm with a sublinear policy regret has a linear external regret, and vice visa. Based on the techniques proposed in \cite{Arora2018PRR}, we prove that such a statement also holds for the online learning problem on the Dd-MDP with chasability. This is the reason why we focus on obtaining vanishing policy regret in the main part of this paper.

\begin{theorem}\label{theorem_policy_regret_external_regret_incomparable}
There exists a $\sigma$-chasable instance of the Dd-MDP so that for any online learning algorithm having a sublinear policy regret on this instance, it cannot guarantee a sublinear external regret on the same instance, and vice visa.
\end{theorem}

\begin{proof}
We start by constructing a deterministic MDP instance and proving that it is a feasible Dd-MDP instance with a constant $\sigma$. 
%The first step is to construct an instance of Dd-MDP and prove that it satisfies the domination condition and the online chasing condition of dominant states by taking $\sigma = m$. 
\OriginalDisplay
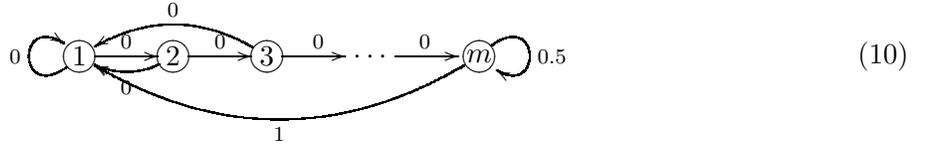
\begin{equation}
\xymatrix{
*+[o][F]{1} \ar[r]^{0} \ar@(dl, ul)^{0} & *+[o][F]{2} \ar[r]^{0}  \ar@/^/[l]^{0} & *+[o][F]{3} \ar[r]^{0}  \ar@/_1pc/[ll]_{0} & \cdots \ar[r]^{0} & *+[o][F]{m} \ar@/^2pc/[llll]^{1} \ar@(ur, dr)^{0.5}
} \label{formula_impossibility_result}
\end{equation}

Consider the deterministic MDP instance with $m > 2$ states in Eq.~\eqref{formula_impossibility_result}, 
\OriginalDisplayEnd
\MultiLineSplit
Consider a deterministic MDP instance with $m > 2$ states, 
\MultiLineSplitEnd
where $m$ is a constant independent of $T$. Each state in this instance is labeled with a distinct integer in $[m]$. The action set contains two actions, which are denoted by \textsc{Forward} and \textsc{Backward}, respectively. These two actions are feasible for every state. The state transition functions $g_{t}$ and reward functions $f_{t}$ are fixed for all $t \in [T]$ as follows.
\begin{align*}
g_{t}(s, x) =& \begin{cases}
s + 1 & \text{if } s < m \bigwedge x = \textsc{Forward} \\
m & \text{if } s = m \bigwedge x = \textsc{Forward} \\
1 & \text{if } x = \textsc{Backward}
\end{cases} \, , \\
f_{t}(s, x) =& \begin{cases}
0.5 & \text{if } s = m \bigwedge x = \textsc{Forward} \\
1 & \text{if } s = m \bigwedge x = \textsc{Backward} \\
0 & \text{otherwise}
\end{cases} \, .
\end{align*}
For any target policy $\gamma \in \Gamma$, initial time $\tinit$ and any initial state $\sinit$, let $k = m - \sinit$, and $\lbrace \hat{x}_{t} \rbrace_{t\geq \tinit}$ be a sequence of actions so that
\begin{equation*}
\hat{x}_{t} = \begin{cases}
\textsc{Forward} & \text{if } t \leq \tinit + k - 1 \\
\textsc{Backward} & \text{otherwise}
\end{cases} \, .
\end{equation*}
This sequence of actions are trivially feasible. For any $\tau \leq \tinit + k - 1$, it is easy to see that
\begin{equation*}
\sum_{t = \tinit}^{\tau} f_{t}(\FPS{t}{\gamma}, \FPA{t}{\gamma}) - f_{t}(\hat{s}_{t}, \hat{x}_{t}) \leq m - 1 \, ,
\end{equation*}
where $\lbrace \hat{s}_{t} \rbrace_{t \geq \tinit}$ is a sequence of states defined in a similar with with Eq.~\eqref{formula_states_visited_in_local_online_learning}. By the setting of the state transition function, we have $\hat{s}_{\tinit + k} = m$. Let $t' = \underset{t \geq \tinit + k}{\argmin} \FPA{t}{\gamma} = \textsc{Backward}$. Then
\begin{equation*}
\sum_{t = \tinit + k}^{t'} f_{t}(\FPS{t}{\gamma}, \FPA{t}{\gamma}) - f_{t}(\hat{s}_{t}, \hat{x}_{t}) \leq 0 \, ,
\end{equation*}
and for every $t > t'$, 
\begin{equation*}
f_{t}(\FPS{t}{\gamma}, \FPA{t}{\gamma}) = f_{t}(\hat{s}_{t}, \hat{x}_{t}) 
\end{equation*}
because in such a case $\FPS{t}{\gamma} = \hat{s}_{t}$ and $\FPA{t}{\gamma} = \hat{x}_{t}$ always hold. Putting the three formulas above together, it is proved that this Dd-MDP instance is $\sigma$-chasable with $\sigma = m - 1$.

Let the number of {\Steps} that an arbitrary algorithm performs the actions \textsc{Forward} and \textsc{Backward} at the state $m$ be $k$ and $k'$, respectively. Note that each time an algorithm performs \textsc{Backward} at the state $m$, then it needs to take at least $m - 1$ {\Steps} to go back to the state $m$. It implies that $k + m \cdot k' \leq T$. The total reward obtained by this algorithm is
\begin{equation}
\frac{1}{2}k + k' \leq \frac{1}{2}k + \frac{1}{m}(T - k) \, . \label{formula_total_reward_and_policy_regret}
\end{equation}
Since the total reward by repeating a fixed policy $\gamma_{F}$ that maps every state to \textsc{Forward} is at least $\frac{1}{2}(T - m)$, the policy regret is at least $(\frac{1}{2} - \frac{1}{m})(T - k) - \frac{m}{2}$. Therefore, if the policy regret is sublinear in $T$, we have $k = T - o(T)$. Now, consider another policy $\gamma_{B}$ that maps every state to the action \textsc{Backward}. We have 
\OriginalDisplay
\begin{equation*}
\sum_{t = 1}^{T}f_{t}(s_{t}, \gamma'(s_{t})) - \sum_{t = 1}^{T}f_{t}(s_{t}, x_{t}) \; \geq \; \Big(1 - \frac{1}{2} \Big) \cdot k \, ,\nonumber
\end{equation*}
\OriginalDisplayEnd
\MultiLineSplit
$\sum_{t = 1}^{T}f_{t}(s_{t - 1}, \gamma_{B}(s_{t - 1})) - \sum_{t = 1}^{T}f_{t}(s_{t - 1}, x_{t}) \geq  \Big(1 - \frac{1}{2} \Big) \cdot k$, 
\MultiLineSplitEnd
which implies that the external regret is linear in $T$.

Now consider an arbirary algorithm whose external regret is sublinear in $T$. Then, the total reward of this algorithm is at most $\frac{T}{m} + o(T)$, because otherwise it can still be inferred from Eq.~\eqref{formula_total_reward_and_policy_regret} that $k$ is linear in $T$, which leads to a linear external regret. Recall that the total reward of repeating the policy $\gamma_{F}$ is $(T - m)/2$. Therefore, the policy regret is linear in $T$.
\end{proof}

\section{Bandit Setting}
\label{section:bandit}

In Section \ref{section:HOLA}, we investigate the online learning problem on Dd-MDPs under the full information setting, which means that for each {\Step} $t$, both the state transition function $g_{t}(s, x)$ and reward function $f_{t}(s, x)$ selected by the adversary are completely revealed to the decision maker after the decision maker chooses a (randomized) action $x_{t}$. In this part, we consider the \emph{bandit} setting, where at each {\Step} $t$, the decision maker only knows the actual reward she receives, $f_{t}(s_{t}, x_{t})$, with the state transition function $g_{t}(\cdot, \cdot)$. Obtaining vanishing regret for Dd-MDPs under the bandit setting requires a stronger condition than $\sigma$-chasability, which is defined in the following subsection.

\subsection{Stateless Chasability}
\label{sec:stateless_chaseability}
%%%%%%%%%%%%%%%%%%%%%%%%%%%%%%%%%%%%%%%%%%%%%%%%%%%%%%%%%%%%%%%%%%%%%%%%%%%%%%
We say that an instance of Dd-MDP satisfies the \emph{stateless chasability} condition for some parameter $\sigma > 0$ if there exists a chasing ongoing oracle $\mathcal{O}^{\mathtt{Chasing}}$ which not only guarantees that $\mathtt{CR} \leq \sigma$, but also ensures that for any target policy $\gamma$ and any initial state $\tinit$, the cumulative reward obtained by taking the generated actions $\lbrace \hat{x}_{t} \rbrace_{t \geq \tinit}$ does not depend on the initial state $\sinit$. More formally, let $\sinit, \sinit'$ be two arbitrary initial states, and $\lbrace \hat{x}_{t} \rbrace_{t \geq \tinit}, \lbrace \hat{x}_{t}' \rbrace_{t \geq \tinit}$ be two sequence of actions generated by the chasing ongoing oracle with starting from $\sinit$ and $\sinit'$, respectively. The stateless chasability condition requires that for any $\tfinal \geq \tinit$
\begin{equation*}
\sum_{t \in [\tinit, \tfinal]} \mathbb{E}\big[f_{t}(\hat{s}_{t}, \hat{x}_{t}) \big] = \sum_{t \in [\tinit, \tfinal]} \mathbb{E}\big[ f_{t}(\hat{s}_{t}', \hat{x}_{t}') \big] \, ,
\end{equation*}
where $\hat{s}_{t}$ and $\hat{s}_{t}'$ are defined in a similar way with Eq.~\eqref{formula_states_visited_in_local_online_learning}. A chasing ongoing oracle is said to be applicable to the bandit setting if its decision on each action $\hat{x}_{t}$ for $t \geq \tinit$ only depends on $\sinit$, $\lbrace f_{t'}(\hat{s}_{t'}, \hat{x}_{t'}) \rbrace_{t' \in [\tinit, t - 1]}$ and $\lbrace g_{t'}(\cdot, \cdot) \rbrace_{t' \in [\tinit, t - 1]}$.

%It is easy to see that this requirement is satisfied by the oracle we constructed for DRACC in Theorem \ref{theorem_chasing_condition}. In particular, it can be inferred from the proof for Theorem \ref{theorem_chasing_condition} that starting from any initial state $\sinit$, the reward obtained by taking the actions $\lbrace \hat{x}_{t} \rbrace_{t}$ up to some {\Step} $\tfinal \geq \tinit$ is always equivalent to
%\begin{equation*}
%\sum_{t \in [\tinit, \tfinal]: (\FAVitems{t}{\FPprice{t}{\gamma}} \neq \emptyset) \bigwedge (\FAVitems{t}{\FPprice{t}{\gamma}} \cap A_{\tinit} = \emptyset)} \hat{q}_{t}^{\gamma} \, ,
%\end{equation*} 
%which is independent of $\sinit$. Furthermore, this oracle is applicable to the bandit setting because it computes each $\hat{\PriceVec}_{t}$ with only using $A_{\tinit}$ and $\lbrace g_{t'}(\cdot, \cdot) \rbrace_{t' \in [\tinit, t - 1]}$.

%%%%%%%%%%%%%%%%%%%%%%%%%%%%%%%%%%%%%%%
\subsection{Multiarmed Bandit Problem}
\label{section:multiarmed_bandit}
To develop vanishing-regret algorithms for $\sigma$-chasable Dd-MDPs under the bandit setting, we utilize technical tools that are related to the \emph{Multiarmed Bandit Problem (MBP)}~\cite{Auer2002NMB}. Using a blackbox algorithm for this problem, \Cref{subsecion_C_and_S_bandit} shows how to obtain vanishing regret for our problem.

In MBP, there is a set of arms $\Gamma$, and  $\Psi \in\mathbb{N}$ rounds. At each {\Step} $\psi \in [\Psi]$, an adversary specifies a reward function $F_{\psi}: \Gamma \mapsto [0, 1]$, which is unknown to the online algorithm at the beginning of this round. Simultaneously, the algorithm chooses an action $\gamma_{\psi} \in\Gamma$.  Then the reward $F_{\psi}(\gamma_{\psi})$ obtained by the algorithm is revealed. The goal of the algorithm is to pick a sequence of actions $\gamma_1,\ldots, \gamma_\Psi$ in an online fashion to maximize $\mathbb{E}\Big[\sum_{\psi \in [\Psi]}F_{t}(\gamma_{\psi}) \Big]$. The regret is defined to be
\begin{equation*}
\max_{\gamma \in \Gamma}\sum_{\psi \in [\Psi]}F_{\psi}(\gamma) - \sum_{\psi \in [\Psi]}\mathbb{E}\Big[F_{\psi}(\gamma_{\psi})\Big] \, . 
\end{equation*}

\begin{theorem}[\cite{Audibert2009MPA}]\label{theorem_INF_regret}
There exists an algorithm \emph{Implicitly Normalized Forecaster} (\texttt{INF}) for MBP whose regret is bounded by $O\left(\sqrt{|\Gamma| \cdot \Psi}\right)$.
\end{theorem}

%%%%%%%%%%%%%%%%%%%%%%%%%%%%%%%%%%%%%%%
\subsection{Decision Making Algorithm: Chasing \& Switching in Fixed-Length Periods}
\label{subsecion_C_and_S_bandit}
%%%%%%%%%%%%%%%%%%%%%%%%%%%%%%%%%%%%%%%

We now present our decision making (DM) algorithm for the $\sigma$-chasable Dd-MDP problems under the bandit setting. Our DM algorithm \emph{Chasing and Switching in Fixed-Length Periods} (C\&S-FLP) requires blackbox accesses to a chasing ongoing oracle that is applicable to the bandit setting and Algorithm \texttt{INF} for MBP, where the action set of MBP is set to be the collections of policies $\Gamma$ in Dd-MDP. Algorithm \texttt{INF} runs over consecutive periods of $\tau$ {\Steps} for some $\tau > \sigma$, while the last period is allowed to have less than $\tau$ {\Steps}. At the beginning of each period $\psi$, C\&S-FLP invokes Algorithm \texttt{INF} to choose a policy $\gamma_{\psi}$ from $\Gamma$. Then it starts a new run of the chasing ongoing oracle $\mathcal{O}^{\mathtt{Chasing}}$ with $\gamma_{\psi}$ as the target policy, $s_{(\psi - 1)\tau + 1}$ as the initial state $\sinit$ and $(\psi - 1)\tau + 1$ as the initial time $\tinit$. Then C\&S-FLP takes the sequence $\lbrace \hat{x}_{t} \rbrace_{t \in [(\psi - 1)\tau + 1, \min \lbrace \psi\cdot \tau, T \rbrace ]}$ of actions generated by $\mathcal{O}^{\mathtt{Chasing}}$ throughout the current period $\psi$, and send the reward $f_{t}(s_{t}, x_{t})$ and state transition function $g_{t}(\cdot, \cdot)$ to $\mathcal{O}^{\mathtt{Chasing}}$ after performing each $\hat{x}_{t}$ as the feedback. After the reward of the last step $t = \min\lbrace T, \psi\tau \rbrace$ of the current period is received, C\&S-FLP computes $F_{\psi}(\gamma_{\psi}) = \frac{1}{\tau}\sum\limits_{t' = (\psi - 1)\tau + 1}^{t} f_{t'}(s_{t'}, x_{t'})$ and feeds it to Algorithm \texttt{INF} as the reward of $\gamma$ at $t$.

\begin{theorem}\label{theorem_regret_of_C_and_S_FLP}
The regret of C\&S-FLP is bounded by $O\left(\frac{\sigma \cdot T}{\tau} + \sqrt{|\Gamma|T\tau }\right)$.
\end{theorem}

\begin{proof}
Let $\Psi = \left\lceil \frac{T}{\tau} \right\rceil$, and $R(\Psi, |\Gamma|)$ be the regret of Algorithm \texttt{INF}. With the stateless condition of the chasing ongoing oracle, $\lbrace F_{\psi} \rbrace_{\psi \in [\Psi]}$ is a sequence of stateless reward functions that satisfy the condition of Theorem \ref{theorem_regret_of_C_and_S_FLP}. Therefore, we have
\begin{equation*}
\max_{\gamma \in \Gamma}\sum_{\psi \in [\Psi]}F_{\psi}(\gamma) - \sum_{\psi \in [\Psi]}\mathbb{E}\Big[F_{\psi}(\gamma_{\psi})\Big] \leq R(\Psi, |\Gamma|) \, ,
\end{equation*}
which gives that
\begin{equation*}
\tau \cdot \max_{\gamma \in \Gamma}\sum_{\psi \in [\Psi]}F_{\psi}(\gamma) - \sum_{t \in [T]}\mathbb{E}\Big[f_{t}(s_{t}, x_{t})\Big] \leq \tau \cdot R(\Psi, |\Gamma|) \, .
\end{equation*}
By the definition of $\mathtt{CR}$, for each period $\psi$ we have
\begin{equation*}
\sum_{t = (\psi - 1)\tau + 1}^{\psi \cdot \tau} f_{t}(\FPS{t}{\gamma}, \FPA{t}{\gamma}) - \tau \cdot \max_{\gamma \in \Gamma}\sum_{\psi \in [\Psi]}F_{\psi}(\gamma) \leq \sigma \, .
\end{equation*}
Therefore, the regret of Algorithm C\&S-FLP is bounded by
\begin{equation*}
\sum_{t \in [T]} f_{t}(\FPS{t}{\gamma}, \FPA{t}{\gamma}) - \sum_{t \in [T]}\mathbb{E}\Big[f_{t}(s_{t}, x_{t})\Big] \leq \sigma \cdot \Psi + \tau \cdot R(\Psi, |\Gamma|) \, .
\end{equation*}
This proposition is proved by plugging Theorem \ref{theorem_INF_regret} into the formula above.
\end{proof}

\begin{corollary}
\label{corollary-bndit}
By taking $\tau = T^{\frac{1}{3}}$, the regret of Algorithm C\&S-FLP is bounded by $O\Big(\sigma T^{\frac{2}{3}} \sqrt{|\Gamma| } \Big)$.
\end{corollary}

\section{Lower Bounds for Online Learning over Dd-MDPs with Chasability}
\label{apx:lower}
In this part, we will prove lower bounds on the regret of online learning algorithms for $\sigma$-chasable Dd-MDP instances under the full-information setting and bandit setting, respectively.

\begin{theorem}\label{theorem_reduction_from_switching_cost}
%Suppose that there exists an online learning algorithm \texttt{ALG} for $\sigma$-Chaseable Dd-MDP instances with state set $\mathcal{S}$ and action set $\mathcal{X}$ under the full-information setting (resp.~bandit setting) with regret $R(T, |\mathcal{X}|)$, then the OLSC problem under the full-information setting (resp.~bandit setting) admits an online learning algorithm with regret $O\Big( R(T, |\mathcal{X}|) \Big)$.
The regret of any online learning algorithm for $1$-chasable Dd-MDP under full-information (resp., bandit) feedback is lower bounded by  $\Omega(\sqrt{T \log |{\Gamma}|})$
(resp., $\Omega(|{\Gamma}|^{1/3}T^{2/3})$).
\end{theorem}

\begin{proof}
We only prove the full information lower bound, and the proof of the bandit version follows the same lines.
Suppose the statement of the theorem does not hold.
Then, there exists an online learning algorithm with regret of
$o(\sqrt{T \log |{\Gamma}|})$
for any $1$-chasable instance of the Dd-MDP problem.
We show how to use this algorithm to design an OLSC algorithm with a unit switching cost whose regret is
$o(\sqrt{T \log |\mathcal{X}|})$,
where $\mathcal{X}$ is the action set of the OLSC instance.
This is in contradiction to the known information theoretic 
$\Omega(\sqrt{T \log |\mathcal{X}|})$
lower bound on the regret of OLSC under the full-information setting \cite{Cesa2006PLG, Freund1995DTG, Littlestone1989WMA}.
(For the bandit version of this proof we use the lower bound of \cite{Dekel2014BSC}).
% To see this, fix an online learning algorithm for $1$-chasable Dd-MDP instances with $T$ {\Steps} and policy collection $\Gamma$ with regret $R(T, |\Gamma|)$.
% We show the OLSC problem with $T$ {\Steps} and action collection $\mathcal{X}$ admits an online learning algorithm with regret $O\Big( R(T, |\mathcal{X}|) \Big)$.

Here is how the reduction works:
Given an OLSC instance with a set $\mathcal{X}$ of actions and a unit switching cost, we construct a Dd-MDP instance with a state $s^{x}$ for each action
$x \in \mathcal{X}$.
An arbitrary state
$s \in S$
is selected to be the initial state $s_{1}$.
Moreover, we set
$\mathcal{X}_{s} = \mathcal{X}$
for every state $s$.
For every
$x \in \mathcal{X}$,
we introduce a policy $\gamma^{x}$ in the policy collection $\Gamma$ of the Dd-MDP, defined so that it maps all states to action
$x \in \mathcal{X}$.
For each round $t$, when the adversary in OLSC specifies a reward function $F_{t}(\cdot)$, we construct the state transition function and reward function in the Dd-MDP by setting
\begin{equation*}
    g_{t}(s, x) = s^{x}
    \qquad \text{and} \qquad
    f_{t}(s, x) = \frac{1}{2}F_{t}(x) + \frac{1}{2} \cdot 1_{s = s^{x}} \, .
\end{equation*}
Obviously, this is a $1$-chasable Dd-MDP instance.
Moreover,
\begin{equation*}
    \max_{\gamma^{x} \in \Gamma}\sum_{t \in [T]} f_{t}(\FPS{t}{\gamma^{x}}, \FPA{t}{\gamma^{x}}) - \sum_{t \in [T]} \mathbb{E}\Big[ f_{t}(s_{t}, x_{t}) \Big] \leq o(\sqrt{T\log|\Gamma|}) = o(\sqrt{T\log|\mathcal{X}|}) \, ,
\end{equation*}
where the inequality is due to the assumed regret bound.
The construction of the Dd-MDP instance ensures that for every $\gamma \in \Gamma$,
\begin{align*}
    \sum_{t \in [T]} f_{t}(\FPS{t}{\gamma^{x}}, \FPA{t}{\gamma^{x}}) \geq&  \, \frac{1}{2} \cdot \sum_{t \in [T]} F_{t}(x) + \frac{1}{2}(T - 1) \, ,\\
    \sum_{t \in [T]} f_{t}(s_{t}, x_{t}) \leq& \frac{1}{2} \sum_{t \in [T]} F_{t}(x_{t}) + \frac{1}{2}T - \frac{1}{2} \sum_{t \in [2, T]} 1_{x_{t} \neq x_{t-1}} \\
\end{align*}
Putting these pieces together, we get
\begin{equation*}
    \frac{1}{2}\bigg(\max_{x \in \mathcal{X}}\sum_{t \in [T]} F_{t}(x) - \sum_{t \in [T]} \mathbb{E}\Big[ F_{t}(x_{t}) \Big] + \sum_{t \in [2, T]}1_{x_{t} \neq x_{t - 1}}\bigg) - 1 \leq o\Big(\sqrt{T\log|\mathcal{X}|} \Big) \, ,
\end{equation*}
and therefore the regret of the OLSC instance is bounded by $o(\sqrt{T\log\lvert\mathcal{X}}\rvert)$, a contradiction.
\end{proof}

% \begin{corollary}
% The regret of any online learning algorithm for $1$-chasable Dd-MDP under full-information (resp. bandit information) is lower bounded by  $\Omega(\sqrt{T \log |{\Gamma}|})$ (resp.~$\Omega(|{\Gamma}$.
% (resp.~$\Omega(|{\Gamma}|^{1/3}T^{2/3})$).
% % Under the full-information setting (resp.~the bandit setting), the lower bound on the regret of online learning algorithms for $\sigma$-Chaseable Dd-MDP with some constant $\sigma > 0$ is $\Omega(\sqrt{T \log |{\Gamma}|})$ (resp.~$\Omega(|{\Gamma}|^{1/3}T^{2/3})$).
% \end{corollary}

%\section{Vanishing Regret and the DRACC Parameters}
%\input{tex/assumptions.tex}

\end{document}